\documentclass[a4paper,twocolumn,11pt,accepted=2022-11-17]{quantumarticle}
 
\pdfoutput=1

\usepackage[english]{babel}
\usepackage[utf8]{inputenc}
\usepackage[T1]{fontenc}
\usepackage[pdftex, pdftitle={Article}, pdfauthor={Author}]{hyperref} % For hyperlinks in the PDF
\usepackage{braket}
 
\usepackage{amssymb,bbm}
\usepackage{amsmath}

\newcommand{\CC}{\mathbb{C}}
\newcommand{\FF}{\mathbb{F}}

\newcommand{\lemacmat}{\mathcal M}  % Linear Equation macmat
\newcommand{\blemacmat}{\mathcal B} % Boolean L E macmat
\newcommand{\macmat}{\hat{\lemacmat}}
\newcommand{\bmacmat}{\hat{\blemacmat}}
\newcommand{\bvec}{\vec{b}}
\newcommand{\s}{a}  % solution of the original equations
\newcommand{\svec}{\vec{y}} % {\alpha}  % the vector version in the Macaulay linear system
\newcommand{\sparsity}[1]{\# #1}
\DeclareMathOperator{\maxdeg}{maxdeg}
\newcommand{\tqlscnnospace}{tQLScn} %{truncated QLS condition number}
\newcommand{\tqlscn}{\tqlscnnospace\space}  %{truncated QLS condition number}

\usepackage{latexsym}
\usepackage{amsthm}
\usepackage[capitalise]{cleveref}
\usepackage[usenames,dvipsnames]{color}
\usepackage{hyperref}\hypersetup{final}
\usepackage{enumitem}
\usepackage{url}
\usepackage{algorithm}
\usepackage{algorithmic}
\usepackage{systeme}
\usepackage{appendix}
\hypersetup{pdfpagemode=UseNone}

\newtheorem{theorem}{Theorem}[section]
\newtheorem{lemma}[theorem]{Lemma}

\newtheorem{corollary}[theorem]{Corollary}
\newtheorem{definition}[theorem]{Definition}

\newtheorem{problem}[theorem]{Problem}

\theoremstyle{definition}
{

}

\newcommand{\floor}[1]{\left\lfloor #1 \right\rfloor}

\newcommand{\tr}{\operatorname{Tr}}

\newcommand{\ip}[2]{\left\langle #1 , #2\right\rangle}

\mathchardef\mhyphen="2D

\newcommand{\bracket}[2]{\langle #1 | #2 \rangle}

\newcommand{\calF}{\mathcal{F}}

\newcommand{\C}{\mathbb{C}}

\def\({\left(}
\def\){\right)}

\def\G{\mathbb{G}}

\usepackage{graphicx}
\definecolor{greenn}{rgb}{0,0.8,0.2}
\definecolor{bluue}{rgb}{0.3,0,0.7}

\usepackage{ifdraft}
\ifdraft{\newcommand{\authnote}[3]{{\color{#3}{\text{ \bf #1:}} #2}}}{\newcommand{\authnote}[3]{}}

\newcommand{\anote}[1]{\authnote{András}{#1}{red}}
\newcommand{\jnote}[1]{\authnote{Jianqiang}{#1}{blue}}

\newcommand{\eps}{\varepsilon}

\newcommand{\nrm}[1]{\left\lVert #1 \right\rVert}
\newcommand{\bigO}[1]{\mathcal{O}\left( #1 \right)}
\newcommand{\bigOt}[1]{\widetilde{\mathcal{O}}\left( #1 \right)}
\newcommand\polylog[1]{\mathrm{poly}\log\left( #1 \right)}

\usepackage{mmacells}

\mmaDefineMathReplacement[≤]{<=}{\leq}
\mmaDefineMathReplacement[≥]{>=}{\geq}
\mmaDefineMathReplacement[≠]{!=}{\neq}
\mmaDefineMathReplacement[→]{->}{\to}[2]
\mmaDefineMathReplacement[⧴]{:>}{:\hspace{-.2em}\to}[2]
\mmaDefineMathReplacement{∉}{\notin}
\mmaDefineMathReplacement{∞}{\infty}

\mmaSet{
	morefv={gobble=2},
	linklocaluri=mma/symbol/definition:#1,
	morecellgraphics={yoffset=0mm},
	leftmargin=11.5mm,
	labelsep=1mm,
}

\begin{document}
\title{Limitations of the Macaulay matrix approach for using the HHL algorithm to solve multivariate polynomial systems}
\author{Jintai Ding}
\affiliation{University of Cincinnati, OH, USA }
\email{jintai.ding@gmail.com}
%\orcid{0000-0002-2445-2701}
 \thanks{Yau Math.\ Sci.\ Center, Tsinghua University}
 \thanks{Ding Lab, Beijing Institute of Mathematical Sciences and Applications}

\author[]{Vlad Gheorghiu}
\email{vlad.gheorghiu@uwaterloo.ca}
\affiliation{Institute for Quantum Computing / Dept.\ of Combinatorics \& Optimization, University of Waterloo, ON, Canada}
\thanks{softwareQ Inc.}

\author{Andr\'{a}s Gily\'{e}n}
\affiliation{Institute for Quantum Information and Matter, Caltech, Pasadena CA, USA} \thanks{Alfréd Rényi Institute of Mathematics, Budapest, Hungary  }
\email{ gilyen@renyi.hu}

\author{Sean Hallgren}
\affiliation{Department of Computer Science and Engineering, Pennsylvania State University, PA, USA }
\email{hallgren@cse.psu.edu}
\author{Jianqiang Li}
\affiliation{Department of Computer Science and Engineering, Pennsylvania State University, PA, USA}
\email{ jxl1842@psu.edu}

\date{07-21-2023}

\maketitle
\vspace{-0.7mm}
\begin{abstract}
Recently Chen and Gao~\cite{ChenGao2017} proposed a new quantum algorithm for Boolean polynomial system solving, motivated by the cryptanalysis of some post-quantum cryptosystems. The key idea of their approach is to apply a Quantum Linear System (QLS) algorithm to a Macaulay linear system over $\mathbb{C}$, which is derived from the Boolean polynomial system. The efficiency of their algorithm depends on the condition number of the Macaulay matrix. In this paper, we give a strong lower bound on the condition number as a function of the Hamming weight of the Boolean solution, and show that in many (if not all) cases a Grover-based exhaustive search algorithm outperforms their algorithm. Then, we improve upon Chen and Gao's algorithm by introducing the Boolean Macaulay linear system over $\mathbb{C}$ by reducing the original Macaulay linear system. This improved algorithm could potentially significantly outperform the brute-force algorithm, when the Hamming weight of the solution is logarithmic in the number of Boolean variables. 

Furthermore, we provide a simple and more elementary proof of correctness for our improved algorithm using a reduction employing the Valiant-Vazirani affine hashing method, and also extend the result to polynomial systems over $\mathbb{F}_q$ improving on subsequent work by Chen, Gao and Yuan \cite{ChenGao2018}. We also suggest a new approach for extracting the solution of the Boolean polynomial system via a generalization of the quantum coupon collector problem \cite{arunachalam2020QuantumCouponCollector}. 
\end{abstract}

%\text{Keywords: Quantum cryptanalysis, HHL, multivariate polynomial equations}

\section{Introduction}\label{sct::intro}

Solving systems of multivariate polynomial equations is a fundamental problem that is NP-complete even when the polynomials are restricted over $\mathbb{F}_2$. The problem can be reduced to solving an exponential number of linear equations via the so-called Macaulay matrix, which holds coefficients of linear equations that come from the input polynomials, and multiples of them (multiplying each polynomial by each monomial up to a certain degree). Each monomial is represented by a new variable, recasting the polynomial equations and their multiples as linear equations. The usual classical approach to solve a polynomial system is based on computing the Gr\"obner basis of the corresponding polynomial ideal by triangularizing the Macaulay matrix. There is a vast literature on characterizing and improving the complexity of solving various types of polynomial systems using the Macaulay matrix~\cite{ars2004comparison,caminata2017solving,courtois2000efficient,ding2013solving,diem2004xl,perret2016bases,wiesinger2015grobner}.   
 
In quantum computing, the HHL~\cite{harrow2009QLinSysSolver} Quantum Linear System (QLS) algorithm outputs a quantum state $\ket{x}$ such that $\mu A\ket{x} = \ket{b}$ for an exponentially large matrix $A$ with certain properties, and a quantum state $\ket{b}$, in time $\bigOt{\kappa^2s^2}$,\footnote{We denote $\bigO{T\cdot\polylog{T}\cdot\mathrm{poly}(1/\eps)}$ by $\bigOt{T}$, where $\eps$ is the required precision of the solution.} where $\kappa$ is the \emph{condition number} of $A$, $\mu$ is a normalization factor, and $s$ is the sparsity of the matrix $A$, while state-of-the-art QLS algorithms~\cite{ambainis2010VTAA,childs2015QLinSysExpPrec,chakraborty2018BlockMatrixPowers,gilyen2018QSingValTransf,subasi2019QAlgSysLinEqsAdiabatic,lin2019OptimalQEigenstateFiltering} have complexity $\bigOt{\kappa s}$. Although, the QLS algorithm is $\mathsf{BQP}$-complete~\cite{harrow2009QLinSysSolver}, meaning that it captures all essential features of quantum computing, a natural ``killer-application'' is still to be discovered -- showing the difficulty of finding a practically interesting problem instance that satisfies all stringent conditions. For example, to efficiently solve the classical equation $Ax=b$ using the original HHL algorithm, where implicit access is given to an exponentially large matrix $A$ and $b$, the following must be satisfied: the state $\ket{b}$ can be efficiently prepared, the sought data can be efficiently extracted from the output state $\ket{x}$, and the matrix $A$ should be sparse and well-conditioned~\cite{aaronson2015ReadTheFinePrint}. 

Chen and Gao~\cite{ChenGao2017} made an interesting connection between the exponential size Macaulay matrix and the HHL algorithm. While they use Gr\"obner bases in their proof of correctness, they do not explicitly compute the Gr\"obner basis and instead use the HHL algorithm to solve the exponentially large system of linear equations resulting from the Macaulay matrix. They show that the access requirements that usually cause so much trouble, can all be resolved for this application, namely: they can efficiently compute the entries of an appropriate sparse matrix $A$, prepare $\ket{b}$, and extract the answer from $\ket{x}$.  However, a major question was left open: what is the condition number of the matrices, driving the running time?  Intuitively, for worst case instances of polynomial systems, the condition number of the resulting matrix should be large because the approach can solve NP-complete problems. This being said, the analysis of the condition number was left open, both in general, and for special cases such as breaking cryptosystems which have distributions over specific problem instances that might be easier than the worst case. Therefore, the algorithm of Chen and Gao~\cite{ChenGao2017} together with the follow-up work of Chen, Gao, and Yuan \cite{ChenGao2018} presented a potential quantum threat on multivariate cryptosystems. However, there was no consensus on the strength of this potential quantum attack, as its cryptanalysis was wide open.

In this paper we prove an exponential lower bound on the condition number $\kappa$ of the matrix $A$ related to the Boolean polynomial system, which shows that the quantum algorithm in~\cite{ChenGao2017} takes exponential time in the worst case.  We also give a Grover-based brute-force search algorithm that outperforms their quantum algorithm for solving Boolean polynomial systems when there is a unique solution or all solutions have the same Hamming-weight. Specifically, in the unique solution case we give a simple proof that the condition number $\kappa$ is $\Omega((3n)^{h/2})$, where $h$ is the Hamming weight of the solution to the original $n$-variable Boolean polynomial system. Meanwhile, a simple Grover-based brute-force search algorithm over the possible assignments to the variables takes time $\bigO{\sqrt{\binom{n}{h}}}$, where $ \sqrt{\binom{n}{h}} \leq \left(\frac{en}{h} \right)^{h/2} \leq (3n)^{h/2}$.

In fact, we give ``robust'' lower bounds on the condition number by also considering ``truncated'' QLS algorithms~\cite{harrow2009QLinSysSolver,gilyen2018QSingValTransf}.  
Namely, if the singular-values of $A$ are only inverted on a \emph{well-conditioned} subspace and the overlap of the solution $x$ with such a subspace is large enough, then a ``truncated'' QLS algorithm can provide a sufficiently accurate solution $\tilde{x}$. In order to give a bound on the performance of such ``truncated'' versions of the QLS algorithm, we define the concept of the \emph{truncated QLS condition number} $\kappa_b(A):=\nrm{A}\nrm{A^+ b}/\nrm{b}$,\footnote{$A^+$ stands for the (Moore-Penrose) pseudoinverse of $A$, and $\nrm{\cdot}$ for the $\ell_2$ norm of vectors and for the corresponding induced operator norm, i.e., the spectral norm.} which is also a lower bound on $\kappa=\nrm{A}\nrm{A^+}$. All of our lower bounds also apply to the truncated QLS condition number, ruling out further improvement by truncated QLS algorithms. These results provide strong evidence that the quantum algorithm of~\cite{ChenGao2017} (at least in its original form) does not present a fatal cryptanalytic threat, and give generic tools for analyzing the strength of individual cryptosystems against this type of quantum attack. 

Finally, we refine Chen and Gao's algorithm to the point that our lower bound does not always rule out the possibility of a superpolynomial quantum speedup even for unique solutions. In particular, the lower bound changes from $(3n)^{h/2}$ to $2^{h/2}$ on our refined algorithm, so for $h=\Theta(\log n)$ the lower bound is only a polynomial, while the brute-force algorithm takes quasi-polynomial time.  Thus, it is conceivable that the condition number is upper bounded by $\text{poly}(n)$ for some set of interesting input equations, potentially yielding a superpolynomial speedup.  We leave it open to find a problem instance whose associated Macaulay matrix has a small enough condition number so that the running time of our refined quantum algorithm gives a speedup over the best classical or Grover-based algorithm.  Such an example could result in a new type of quantum speedup and one that uses the HHL algorithm in a novel way.

The core ingredient of our refined algorithm is to show that the Macaulay matrix can be simplified to what we call the Boolean Macaulay matrix over $\CC$ by exploiting that the input consists of quadratic polynomials over $\CC$, but restricted to 0/1 solutions. The Boolean Macaulay matrix is a submatrix of the original Macaulay matrix that can be obtained via Gaussian elimination over $\CC$. This construction of the Boolean Macaulay matrix over $\CC$ is different from the Boolean Macaulay matrix over $\mathbb{F}_2$ as defined in~\cite{bardet2013complexity} since they are over different fields. This matrix preserves the solution set while its size is much smaller compared to the original Macaulay matrix -- ultimately leading to a smaller lower bound $\Omega(2^{h/2})$ on the condition number. 
%For Hamming weight $h \leq \left\lfloor n/2 \right\rfloor$, we have $2^{h/2}\leq \left( \frac{n}{h} \right)^{h/2} \leq \sqrt{\binom{n}{h}}$. Therefore, (by flipping the value if $h \geq \left\lfloor n/2 \right\rfloor$) this refined algorithm might be able to outperform the Grover-based brute-force search approach.  

The correctness of our refined algorithm can be shown via the equivalence between the Boolean Macaulay linear system and the Macaulay linear system, where the correctness of the latter has been proven in \cite{ChenGao2017}. For completeness, we also provide a simple self-contained proof of correctness for our improved algorithm in Appendix~\ref{append:simple}, which is more elementary than that of the original algorithm proposed by Chen and Gao~\cite{ChenGao2017}, since our proof does not require Gr\"obner bases. Instead, our proof combines a special case of the reduction in~\cite{ChenGao2018} with the Valiant-Vazirani affine hashing method, reducing any Boolean polynomial system with more than one solution to one that has a unique solution. For a Boolean Macaulay linear system that has a unique solution, we also provide an alternative approach for extracting the Boolean solution of the corresponding Boolean polynomial system from the output quantum state, i.e., the normalized monomial solution vector of the Boolean Macaulay linear system. Specifically, we reformulate this problem as a generalization of the quantum coupon collector problem, and prove that $O(\log n)$ iterations suffice for extracting a solution, whereas Chen and Gao's algorithm uses $O(n)$ iterations. On the other hand, the affine hashing reduction introduces $\bigO{n}$ extra rounds, so the total number of iterations in our algorithm can be bounded as $\bigO{n \log n}$.

\section{Quantum algorithms for solving polynomial systems}
 
 Chen and Gao~\cite{ChenGao2017} proposed using the HHL algorithm to solve a Boolean polynomial system via solving an exponentially large linear system of equations. Now we discuss the two main parameters that appear in the complexity of the QLS algorithm, and their relevance in our case. 
  
\begin{enumerate}
  \item[$\kappa_b(A):$] The \emph{truncated QLS condition number} (\tqlscnnospace) $\kappa_b(A)$ of the QLSP $Ax=b$ is an important parameter related to the time-complexity of ``truncated'' QLS algorithms.\footnote{Note that the truncated QLS condition number gives a lower bound on the performance of truncated QLS alorithms, but it does not characterize their complexity, i.e., there might not exist any truncated QLS algorithm with complexity matching the truncated QLS condition number.} (For simplicity let us assume without loss of generality that $\nrm{A}\leq 1$ and $\nrm{b}=1$.)
  We use a simple Markov-type inequality showing that inverting $A$ via a truncated variant of the QLS algorithm, with (condition number) truncation much below $\kappa_b(A)$, must give a highly inaccurate solution. Indeed, let $S$ be a subspace, which is spanned by right singular vectors of $A$ that have singular value at least $c/\kappa_b(A)$ for some $c>1$. Let $\Pi_S$ be the orthogonal projector on $S$, then $\nrm{\Pi_S x} = \nrm{\Pi_S A^+ b}\leq \nrm{\Pi_S A^+} \nrm{b} \leq  \kappa_b(A)/c$. On the other hand $\nrm{x} = \nrm{A^+ b}\geq \nrm{A} \nrm{A^+ b}/ \nrm{b} =  \kappa_b(A)$, thereby the overlap of $x$ with the subspace $S$ must be relatively small for large $c$. This argument shows, that giving a lower bound on the truncated QLS condition number makes our results stronger, as it does not only bound the complexity of standard QLSP solvers, but also lower-bounds the complexity of fine-tuned truncated variants.
	
 \item[$s:$] In order to efficiently solve the Quantum Linear
          System Problem (QLSP), we need a succinct representation of
          the vector $b$ and the matrix $A$. In our case both
          $b$ and $A$ are $s$-sparse for some polynomially large $s$,
          i.e., they have at most $s$ nonzero entries in every column
          and row. In order to utilize their sparsity, we also need
          to be able to efficiently compute the locations and the
          values of the nonzero entries; this is easy to do in our
          case, since the vector $b$ is sparse and the (Boolean)
          Macaulay matrix $A$ has a quasi-Toeplitz structure.  
\end{enumerate}

More precisely, it suffices to have efficient (quantum) circuits computing the locations and values of the nonzero elements of $A$, allowing us to perform the transformations 

%{\color{blue}{Why ``in-place" transformation in eqn.(1)}}
\begin{align}
\ket{i,k}&\longrightarrow\ket{i,\bar{c}(i,k)}\label{eq:prepA1}\\
\ket{i,j,0}&\longrightarrow\ket{i,j,\bar{A}_{ij}},\label{eq:prepA2}
\end{align}
where $i$ labels the row indices of the matrix $\bar{A}$ standing for $A$ and $A^T$, $k$ labels the nonzero entries of $\bar{A}$ (which is assumed to be $s$-sparse), $\bar{c}(i,k)$ represents the column index of the $k$-th nonzero entry of the matrix $\bar{A}$ in row $i$, and $\ket{0}$ in~\eqref{eq:prepA2} represents a (large enough) ancillary system in which the matrix element $\bar{A}_{ij}$ can be stored (as a bit-string with sufficient precision so that errors can be neglected). Note that the transformations \eqref{eq:prepA1}-\eqref{eq:prepA2} are essentially only used to build an efficient quantum circuit $C_A$~\cite[Lemma 48]{gilyen2018QSingValTransfArxiv}, which implements a unitary $U$, such that the top left $N\times M$ corner of $U$ equals $A$ -- such a unitary is called a block-encoding of $A$~\cite{chakraborty2018BlockMatrixPowers,gilyen2018QSingValTransf}. Similarly, the sparsity assumption for $b$ is only used in order to build an efficient quantum circuit $C_b$ for preparing the quantum state $\ket{b}:=\left(\sum_{i=1}^Nb_i\ket{i}\right)/\nrm{\sum_{i=1}^Nb_i\ket{i}}$. Thereby, the QLSP problem can also be efficiently solved for non-sparse $b$ or $A$, whenever we can build efficient quantum circuits $C_b$, $C_A$. For example, if we have both $b$ and $A$ stored in a quantum random access memory (qRAM) using an appropriate data structure, then we can build the efficient quantum circuits $C_b$, $C_A$ as described in \cite{kerenidis2016QRecSys,chakraborty2018BlockMatrixPowers,gilyen2018QSingValTransf}.

For completeness, let us mention two additional types of quantum algorithms that have been proposed for solving polynomial systems:
\begin{enumerate}
\item QAOA : In 2002, Burges \cite{burges2002factoring} reformulated RSA factoring problem as a polynomial system solving problem that has a unique solution, which is a special case of Problem~\ref{prob:MQC}. Later, Anschuetz, Olson, Aspuru-Guzik and Cao \cite{anschuetz2018variational} reformulated the polynomial system solving problem as a Local Hamiltonian Problem (LHP) that has a corresponding unique ground state, and they applied the QAOA
algorithm for finding this ground state; the exact complexity of this algorithm is unknown. 

\item Grover : In 2017, Faug\`ere, Horan, Kahrobaei, Kaplan, Kashefi, and Perret \cite{faugere2017fast} presented a quantum version of the classical algorithm in \cite{bardet2013complexity}, that applies Grover search on both the exhaustive search and the consistency check subroutines. Under certain assumptions, the running time of this quantum algorithm is slightly better than $\bigO{2^{n/2}}$, which is the running time of trivial Grover search algorithm. Similarly, Bernstein and Yang \cite{bernstein2018asymptotically} gave a quantum
algorithm (GroverXL) for random polynomial systems over a finite
field $\mathbb{F}_q$.\footnote{The generalization of the Boolean Macaulay matrix method in
\cite{bardet2013complexity} is equivalent to the reduced XL approach that first appeared in \cite{courtois2000efficient} and then defined in \cite{diem2004xl}.}
 
\end{enumerate}
While QAOA is only a heuristic algorithm and Grover search only represents a quadratic speedup, 
the QLS algorithm -- being $\mathsf{BQP}$-complete -- captures the power of quantum computation and
promises rigorous superpolynomial speedups, giving strong motivation for the study of the Macaulay matrix approach.

\section{Reducing polynomial system solving over a finite field \texorpdfstring{$\FF_q$}{F{\scriptsize 2}} to polynomial system solving over \texorpdfstring{$\mathbb{C}$}{C}}\label{sct:reductionToQLPS}

First, we present the $\FF_q=\FF_2$ special case of Chen, Gao and Yuan's approach~\cite{ChenGao2018}.
In this special case there is a bijection between the solution sets of the corresponding polynomial systems over $\FF_2$ and $\CC$.

%{\color{blue} The necessary of the $F_2$ ? }
\begin{problem} {Solve a system of $n$-variate quadratic polynomials with
		Boolean variables over
		$\mathbb{F}_2$.} \label{prob:MQ2} \\
	\indent
	\textbf{Input :}
	$\calF=\{f_1,\dots, f_m\} \subseteq \mathbb{F}_2[x_1,\dots,x_n]$ with
	$\deg(f_i) \leq 2$ for 
	$i=1,2,\dots, m$.\\
	\indent	\textbf{Output :} Find a solution $s \in
	\mathbb{F}_2^n$ such that $f_1 (s)  = \cdots = f_m
	(s) = 0$, when one exists.
\end{problem}

\begin{problem} {Solve a system of $n$-variate quadratic polynomials
		over $\mathbb{C}$, together with the $\FF_2$ field equations that force variables to be
		Boolean.} \label{prob:MQC}\\
	\indent \textbf{Input :}
	$\calF  
 \subseteq \mathbb{C}[x_1,
	\dots,x_n]$ where $\calF=\{{ f_1}, \dots, { f}_m\} $ with
	$\deg(f_i) = 2$ for $i=1,\dots, m$, %and	$\calF_2=\{x_1^2-x_1,\ldots, x_n^2-x_n\}$.\\
	\indent \textbf{Output :} Find an $s \in \{0,1\}^n$ such that
	$ { f_1} (s) = \cdots= { f}_m(s) = 0$ over $\CC$, when one exists.
\end{problem}

Let $\sparsity{\calF}$ denote the maximum number of nonzero terms in any polynomial in $\calF$. Also let $\sparsity{f}$ be a shorthand for $\sparsity{\{f\}}$.

\begin{lemma}\label{lem: polyreduction}
	There is a polynomial-time reduction from Problem~\ref{prob:MQ2} on
	$n$ variables and a set of $m$ equations $\calF$ to Problem~\ref{prob:MQC} on
	$n+m\cdot \lfloor \log_2 \sparsity{\calF} \rfloor $ variables and
	$n+m\cdot (\lfloor \log_2 \sparsity{\calF}\rfloor + 1)$ equations.
\end{lemma}

\begin{proof}
	
	Solving Problem~\ref{prob:MQ2} is equivalent to solving the following
	polynomial system in variables $x_1,\ldots,x_n,$ $z_1,\ldots, z_m$ over
	$\mathbb{C}$:
	\begin{align}
	&\forall i\in [m]\colon &  f_i(x_1,\dots,x_n)  -z_i & = 0, \label{eq:presys}	\\
	&\forall i\in [m]\colon &  z_i/2 & \in \mathbb{Z},  \label{eq:even}\\
	&\forall j\in [n]\colon &  x_j^2 - x_j & = 0. \label{eq:field}
	\end{align}

The field equations~\eqref{eq:field} force each $x_j$ to be 0 or 1, and therefore each $f_i$ evaluates to an even or odd integer. Also, it is easy to see that each integer $z_i$ in equation~\eqref{eq:even} is in $ [0, \sparsity{f_i}]$ for $i=1,\dots m.$ and can be treated as a polynomial.
 Equations~\eqref{eq:presys}-\eqref{eq:even} force $f_i$ to evaluate to an even integer, so it is 0 mod 2.

	Chen, Gao and Yuan~\cite{ChenGao2018} then represent each $z_1,\ldots, z_m$ by the bits in its binary expansion. 
	For each variable $z_i\in [0, \sparsity{f_i}]$ a polynomial is introduced with Boolean variables $y_{ib}$ to represent its value in binary, i.e.,
	$
	z_i = \sum_{b=1}^{\floor{\log_2\sparsity f_i}} 2^{b} y_{ib}, \text{
		and }  y_{ib}^2 - y_{ib} = 0  \text{ for }  b=1,\dots, \floor{\log_2\sparsity{f_i}}.
	$
	
	Substituting the polynomials and Boolean constraints corresponding to each $z_i$ into the polynomial
	equations~\eqref{eq:presys}-\eqref{eq:field}, we get a following polynomial	system $\calF$ over $\mathbb{C}$:
	\begin{align}
		&\forall i\in [m]\colon   f_i(x_1,\dots,x_n) \textstyle -\sum_{b=1}^{\floor{\log_2\sparsity{\calF}}} 2^{b} y_{ib}  = 0, \label{eq:finalsys}	\\
		&\forall i\in [m]\forall b\in [\floor{\log_2\sparsity{\calF}}]\colon   y_{ib}^2 - y_{ib}  = 0, \label{eq:fieldfin1}\\
		&\forall j\in [n]\colon   x_j^2 - x_j  = 0. \label{eq:fieldfin2}
	\end{align}
	
	It is easy to see that there is a bijection between the set of solutions of $\calF \subseteq \mathbb{F}_2[x_1,x_2\dots,x_n]$ and the set of solutions of~\eqref{eq:finalsys}-\eqref{eq:fieldfin2} over $\CC$. 
	On one hand, given a solution $(s_1,\ldots,s_n)$ of $\calF \subseteq \mathbb{F}_2[x_1,x_2\dots,x_n]$, 
	evaluating $f_i(s_1,\ldots,s_n)$ over $\CC$ gives an even number $z_i$,
	and its binary expansion gives the values of the $y_{ib}$ variables.	
	On the other hand, let $(s_j),(t_{ib})$ be a solution
	to~\eqref{eq:finalsys}-\eqref{eq:fieldfin2}.  
	For each $j$, $s_j \in \{0,1\}$	by \eqref{eq:fieldfin2}, and for each $i$, $f_i(s_1,\ldots,s_n) =\sum_{b=1}^{\floor{\log_2\sparsity{\calF}}} 2^{b} y_{ib} $ by~\eqref{eq:finalsys}. Due to \eqref{eq:fieldfin1} this implies that $f_i(s_1,\ldots,s_n)\equiv 0\mod 2$, so $(s_j)$ is a solution of $\calF$.
\end{proof}

Given a set of polynomials $\calF \subseteq \mathbb{F}_q[X]$, Chen,
Gao and Yuan \cite{ChenGao2018} propose an analogous approach for
reducing the polynomial system $\calF$ to a polynomial system
$\calF_\CC = 0$, where $\calF_\CC \subseteq \CC[X,Y]$, in the form of
\cref{prob:MQC}.  However, in general even if $\calF$ has a unique solution,
the constructed polynomial systems $\calF_\mathbb{C}$ may have
multiple solutions.  

Now we present two additional reduction steps that make the polynomial systems in \cref{prob:MQC} easier to handle:
\begin{enumerate}[label=Red\arabic*:\!, ref=Red\arabic*]
\item\label{it:redVV} In order to ensure that the polynomial system $\calF \subseteq \mathbb{C}[X]$ in \cref{prob:MQC} has no more than one solution, we employ the Valiant-Vazirani affine hashing method \cite{valiant1986NPEasyAsDetectingUniqueSols}. Suppose that the polynomial system $\calF$ has $S\in[2^n]$ different solutions. The main idea of the affine hashing method is the following~\cite{bjorklund2019solving}: if one introduces $\floor{\log_2(S)}+2$ random linear equations $\calF_{R}$ with $\calF_{R} \subseteq \mathbb{F}_2[X]$, then they isolate a unique solution with probability at least $\frac{1}{8}$. Even if we don't know the number of solutions a priori, we can loop over all possible values of $\floor{\log_2(S)}\in\{0,1,\ldots,n\}$; making $\bigO{\ln(1/\eps)}$ trials for all possible choices of $\floor{\log_2(S)}$ gives at least one system $\calF_{R} = 0$ with a unique solution with probability at least $1-\eps$. This amounts to $\bigO{n\log(1/\eps)}$ different polynomial systems to check. 
Remember that $\calF_{R} \subseteq \FF_2[X]$ whereas $\calF \subseteq \mathbb{C}[X]$.	By Lemma ~\ref{lem: polyreduction}, we can reduce the new polynomials $\calF_{R}$ to polynomials $\calF_{R\CC} \subseteq \mathbb{C}[X,Y]$, where $Y$ is the set of new variables introduced during the reduction. Finally, if $\calF_{R}$ isolated a unique solution of $\calF$, then the polynomial system $\calF_{R\CC}\cup \calF$ has a unique solution. Thus, without loss of generality, we can always assume that \cref{prob:MQC} has a unique solution. 
	
\item\label{it:redN} Any polynomial system	$\calF=\{f_1, f_2 ,\dots, f_m \} \subseteq \C[X]$ can be rewritten as $\calF'= \{f_1', f_2' ,\dots, f_m' \}$, where $f_1'$ has constant term $-1$, while $f_2',\dots, f_m'$ have no constant terms.
	In case no polynomial in $\calF$ has a constant term, the all-zero vector is a trivial solution. Otherwise, let $c_i$ denote the constant term of $f_i$, and let us assume without loss of generality that $c_1\neq 0$. Then we can simply set $f_1':=-f_1/c_1$, and $f_i':=f_i+c_i f_1'$ for all $i\in\{2,3,\ldots,m\}$.
\end{enumerate}

The above two reductions increase the parameters considered in this paper only moderately. Indeed, \ref{it:redVV} introduces at most $\bigO{n \log(n)}$ new equations and variables, while \ref{it:redN} only affects the number of nonzero terms in the polynomial system. Moreover, \ref{it:redN} increases $\sparsity{\calF}$ and the total number of nonzero terms $\sum \sparsity{f_i}$ by at most a factor of $2$ (for the latter we shall choose $f_1$ to be the polynomial with a nonzero constant term that also has the fewest nonzero coefficients).

\newcommand{\mlmon}{m}
\newcommand{\mlmonb}{m'}
\newcommand{\mlmonc}{m''}
\newcommand{\mlmond}{t}
\newcommand{\mon}{\hat{m}}
\newcommand{\monb}{\hat{m}'}

\section{Macaulay linear systems and their \tqlscn}

In this section we define the Macaulay linear system of a set of
polynomials $\calF \subseteq \CC[x_1,\ldots,x_n]$ and show that 
when $\calF$ has a unique solution, the
condition number of the matrix is $\Omega(\sqrt{\binom{n}{h}})$, where 
$h$ is the Hamming weight of
the solution.  We show that the lower
bound also holds when using max degree instead of total degree in the
definition, as was done in~\cite{ChenGao2017}, showing that their
proposed quantum algorithm for solving polynomial equations by using
the QLS algorithm to solve a Macaulay linear system in general takes time
$\Omega((3n)^{h/2})$. We also show that if there are $t$ different solutions, 
but they have the same Hamming weight $h$, then the above lower bound 
can reduce by at most a factor of $\sqrt{t}$. Finally, we give a formula 
that can be used for giving a lower bound on the condition number
for any number of solutions and present computational evidence that 
this analytical lower bound is exponentially large in terms of the smallest 
Hamming weight among the solutions.

\subsection{Macaulay linear systems}

There is a well-known approach for solving polynomial systems by linearizing them with the help of introducing new latent auxiliary variables. The advantage of this approach is that the problem becomes linear, but the downside is
that the new problem is exponentially large.  The matrix of the resulting linear system is
called the Macaulay matrix.

\begin{definition}\label{def:MacaulayMat}
  The {\em Macaulay matrix $\macmat$ of degree $d$} of
  $\calF=\{f_1,\ldots,f_m\}\subseteq \CC[X]$ is the matrix where each
  row is labeled by a pair of polynomials $(\mon,f)$ and contains
  the corresponding coefficient vector of the polynomial $\mon f$.
  The rows range over all $f\in \calF$ and monomials $\mon$ such
  that $\mon f $ has degree at most $d$.  The columns are labeled by
  the set of monomials in $x_1,\ldots, x_n$ of degree at most $d$ and
  are ordered with respect to a specified monomial ordering.  The
  element in the row corresponding to $(\mon,f)$ and the column
  corresponding to the monomial $\monb$ is the coefficient of $\monb$
  in the polynomial $\mon f$.
\end{definition}

In the above definition one can interpret the degree as either the total degree 
or the max degree (the maximum degree of any variable) of multivariate polynomials 
resulting in different notions of the Macaulay matrix. When it is necessary 
we will always clarify which definition is being used. For example,
\cite{ChenGao2017} uses the max degree, so all references to that paper refer to the max degree version of this definition. 

 In the classical setting, the goal is to compute the Gr\"obner basis from the Macaulay matrix, where the Gr\"obner basis is a set of polynomials $\G=\{g_1, g_2,\dots, g_r\}$ such that for the leading term of any polynomial $f$ in the ideal $I=(\calF) $, there exists a polynomial $g_k \in \G$ such that $LM(g_k)|LM(f)$ \footnote{ Here $LM(f)$ is the leading term of the polynomial $f$}. Note that the size of the Macaulay matrix depends on the selected degree. For a set of $m$ quadratic polynomials with $n$ variables, the degree is approximately lower bounded by $\frac{n}{\sqrt{m}}$ \cite{courtois2000efficient}. When $m=\alpha n$, the degree is upper bounded by $c_{\alpha}n$ for some constant $c_{\alpha}$ \cite{bardet2013complexity}. In the quantum setting, the goal is to compute the monomials up to a certain degree. In this paper, we only provide the upper bound of the degree  that applies to any quadratic polynomials. It might be interesting to check the degree of some special polynomial systems.
 
Row operations on the matrix $\macmat$ correspond to
polynomial addition, subtraction, and scalar multiplication in the
polynomial ideal $\left\langle \calF\right\rangle$
\cite{batselier2013numerical,buchberger2018grobner}, and these
operations preserve the common roots of the system.  Classically, Gaussian
elimination can be performed on this matrix, and the entries can be
read out from the row-reduced matrix~\cite{diem2004xl,courtois2000efficient}. 
However, in the quantum case, we cannot directly do Gaussian elimination on this matrix and 
look at the row-reduced matrix. Instead, Chen and Gao showed that the QLS
algorithm can be used for sampling from nonzero solutions of the
following related linear system.

\begin{definition}\label{def:Macaulay}
  Let $\macmat$ be the Macaulay matrix of a given polynomial system,
  with the last column~$-\bvec$ corresponding to the constant terms of the
  polynomials.  Let $\macmat = [ \lemacmat |\ -\bvec\ ]$.  Then the
  equation
$ \lemacmat \svec = \bvec$
is called the {\em Macaulay linear system}.
\end{definition}

In other words, the Macaulay matrix is the augmented matrix from the
Macaulay linear system $\lemacmat \svec=\bvec$. 
Due to reduction \ref{it:redN} in the previous section, it can be assumed that exactly one of the
input polynomials has a nonzero constant term.
So we may assume without loss of generality that $\bvec = [1 \quad \textbf{0}]^T$, where the 
vector $\bvec$ corresponds to the column vector indexed by the degree~$0$~monomial~$1$.  

Chen and Gao's algorithm applies the QLS
algorithm to output the quantum state $\ket{\hat{y}}$ that can be
measured in order to sample from monomials with nonzero value in a valid assignment corresponding to a solution. We will
lower bound the condition number of $\lemacmat$, which will in turn
lower bound the running time of the proposed algorithm. 

Chen and Gao proposed to set the max degree to $3n$ in the Macaulay linear system, and they showed with this choice if a set of polynomials $\calF$ has a unique solution, then the linear system also has a unique solution~\cite[Lemma 4.1]{ChenGao2017}.
The output state $\ket{\hat{y}}$ of the QLS algorithm then corresponds to this unique solution of the linear system, and the solution of $\calF$ can be efficiently obtained from measuring the state $\ket{ \hat{y} }$. 

Classically, the solving degree of $\calF$ is used, which is at most $n+2$ as shown by Caminata and Gorla~\cite[Theorem 3.26]{caminata2017solving}.
This allows computing the Gr\"obner basis of the polynomial ideal $\left\langle \calF\right\rangle$ via Gaussian elimination of the linear system, and the solution of $\calF$ can be obtained from the Gr\"obner basis.
However, for some polynomial systems, the affine subspace of all solutions of the linear system has no well-understood structure even though $\calF$ has a unique solution. In this case, the QLS algorithm outputs a state $\ket{\hat{y}}$ that corresponds to the smallest $\ell_2$-norm solution of the linear system. In general we don't know how to extract the solution of $\calF$ from such states $\ket{\hat{y}}$.  

When $\calF$ has more than one solution and the max degree is set to be $3n$, the dimension of the affine subspace of all solutions of the linear system equals the number of solutions of $\calF$ minus 1. 
For each solution $\s \in \{0,1\}^n$ of $\calF$, there is a corresponding solution $\hat{y}_{\s}$ of the linear system. Those solutions $\hat{y}_{\s}$ are linearly independent and any solution of the linear system is an affine combination of the solutions $\hat{y}_{\s}$~\cite[Lemma 3.18]{ChenGao2017}. Again, the QLS algorithm outputs a state $\ket{ \hat{y}}$ corresponding to the smallest $\ell_2$ norm vector $\hat{y}$ in this affine subspace.
 
Chen and Gao~\cite{ChenGao2017} showed that $\lemacmat$ is an $\bigO{m\cdot\sparsity{\calF}}$-sparse, row computable matrix and $\bvec$ can be efficiently prepared as a quantum state. 
In particular, assuming $|\calF|=\bigO{\mathrm{poly}(n)}$, they show that the QLS algorithm can be run in time $\bigOt{\mathrm{poly}(n)\kappa(\lemacmat)}$.  
There is strong complexity theoretic evidence that in general running the QLS algorithm requires time $\widetilde{\Omega}(\kappa(\lemacmat))$~\cite{harrow2009QLinSysSolver}, so a lower bound on the condition number also lower bounds the running time.

\subsection{Lower bound on the truncated QLS condition number \texorpdfstring{$\kappa_{\bvec}(\lemacmat)$}{k{\scriptsize b}(\lemacmat)}} 

\newcommand{\onevec}{ \pmb{1}}
\newcommand{\covec}{c}
\newcommand{\dualvalue}{\gamma}

In this section we give a lower bound on the \tqlscn $\kappa_{\bvec}(\lemacmat)$. Since $\kappa_{\bvec}(\lemacmat) \leq \kappa(\lemacmat)$, 
this also implies a lower bound on the time complexity of Chen and Gao's~\cite{ChenGao2017} algorithm. 

In order to prove a lower bound on $\kappa_{\bvec}(\lemacmat)$, it suffices to lower bound the length of the solution vector $\svec = \lemacmat^+ \bvec$, since
	\begin{equation}\label{eq:condStrat}
\kappa_{\bvec}(\lemacmat)=  \nrm{\lemacmat}\frac{\nrm{\lemacmat^+
		\bvec}}{\nrm{\bvec}} \geq \nrm{\lemacmat^+ \bvec} =
\nrm{\svec}.
\end{equation}
Here, the first equality is the definition of $\kappa_{\bvec}$, and the inequality follows from $\nrm{\bvec}=1$, and because
$\nrm{\lemacmat}\geq 1$, as $\lemacmat$ has at least one matrix element which has absolute value at least $1$.\footnote{
For the Macaulay matrix construction,  let $f$ be $x^2-x$, the row $(1,f)$ will have a matrix element of magnitude $1$.} 

In order to understand the length $\nrm{\svec}$  of the solution vector $\svec = \lemacmat^+ \bvec$, let us first study the monomial solution vector $\svec^{(a)}$ corresponding to a binary solution $a$.
For a degree-$d$ Macaulay linear system, a monomial exponent $0\neq e\in {\mathbb N}^n$ is a ``valid'' coordinate of $\svec^{(a)}$ if $e\in\{0,1,\ldots,d\}^n$ (and $\sum e_i \leq d$ for total degree), 
moreover the solution vector satisfies $\svec^{(a)}_e = \Pi_i \s_i^{e_i}$, 
which is $1$ if and only if $\s_i=1$ for all variables $x_i$ in the monomial $\Pi_i x_i^{e_i}$ indexed by $e$. 
If the Hamming weight of $a$ is $h$ and $\lemacmat$ is constructed with max degree, 
then the number of such non-zero coordinates (monomials) is $(d+1)^{h}-1$, thus 
\begin{equation}\label{eq:vecMax}
\nrm{\svec^{(a)}}^2=(d+1)^{h}-1.
\end{equation}
When $\lemacmat$ is constructed with total degree, the number of such non-zero coordinates (monomials) is $\binom{d+h}{h}-1$, so
\begin{equation}\label{eq:vecTot}
\nrm{\svec^{(a)}}^2=\binom{d+h}{h}-1.
\end{equation}

Suppose $\s_1,\s_2,\cdots,\s_t \in \{0,1\}^n$ are the $t$ solutions of
$\calF=~{\calF}_1 \cup~\calF_2$, where $\calF_1=\{{ f_1}, \dots, { f}_m\} $  and
$\calF_2=\{x_1^2-x_1,\ldots, x_n^2-x_n\}$. 
The $t$ solutions $\s_1,\s_2,\cdots,\s_t $ of $\calF$ must be nonzero because the first equation has constant term $b_1=1$.
Let $\svec_1,\svec_2,\cdots,\svec_t$ \footnote{Note that the monomial solution vector corresponding to a binary solution $a_i$ is $\svec^{a_i}$.  Here and in the following discussion of multiple solutions case,  we write $\svec^{a_i}$ as $\svec_i$ for simplicity.
 } be the 0/1 solution vectors of the linear system $\lemacmat \svec = \bvec$ 
under the assignments $\s_1,\s_2,\cdots,\s_t$ respectively. 
When the max degree of the  linear system is set to be $3n$, the affine subspace of all solutions of the linear system is spanned by the monomial solution vectors $\svec_1,\svec_2,\cdots,\svec_t$~\cite[Theorem 3.21 and Lemma 4.1]{ChenGao2017}, but this property might also hold for lower degrees. From now on we assume that degree $d$ is such that the linear system has this property, then $\svec = \lemacmat^{+}\bvec$ has the minimum $\ell_2$ norm in the affine subspace spanned by  $\svec_1,\svec_2,\cdots,\svec_t$.

If all the $t$ solutions $\s_1,\s_2,\cdots,\s_t \in \{0,1\}^n$ have the same Hamming-weight, we can lower bound the length $\nrm{\svec}$ of the  solution vector $\svec=\lemacmat^{+}\bvec$ by the following lemma. 
\begin{lemma}
	\label{lem:LowerboundMFMultiple1}
	Suppose $\svec_1,\svec_2,\cdots,\svec_t$ are vectors with $0$ and $1$ entries such that their Hamming weights are equal. Then every vector in their (complex) affine hull $A$ has length ($\ell_2$ norm) at least $\nrm{\svec_1}/\sqrt{t}$.
\end{lemma}
\begin{proof}
	Every entry of the vectors $\svec_i$ is either $0$ or $1$, therefore $\ip{\svec_i}{\onevec}=\nrm{\svec_i}_1$ (we denote by $\onevec$ the all-$1$ vector). Since the vectors have the same Hamming weight we also have $\nrm{\svec_1}_1=\nrm{\svec_i}_1$ for all $i\in[t]$. Let $\svec$ be any vector in $A$, then $\ip{\svec}{\onevec}=\ip{\svec_1}{\onevec}$, and in particular $\nrm{\svec}_1\geq\ip{\svec}{\onevec}= \nrm{\svec_1}_1=\nrm{\svec_1}_2^2$.
	Let $\nrm{\svec}_0$ denote the support size of $\svec$. Then $\nrm{\svec}_0\leq \sum_{i=1}^{t}\nrm{\svec_i}_0=t\nrm{\svec_1}_0=t\nrm{\svec_1}_1=t\nrm{\svec_1}_2^2$. By the Cauchy-Schwarz inequality we have that 
	$ \nrm{\svec}_1\leq\nrm{\svec}_2\sqrt{\nrm{\svec}_0} \Longrightarrow \nrm{\svec}_2\geq \frac{\nrm{\svec}_1}{\sqrt{\nrm{\svec}_0}}\geq \frac{\nrm{\svec_1}_2^2}{\sqrt{t}\nrm{\svec_1}_2}=\frac{\nrm{\svec_1}_2}{\sqrt{t}}.\qedhere
	$
\end{proof}
 
If the minimum $\ell_2$-norm solution $\svec = \macmat^{+}\bvec$ happens to be a convex combination $\svec =\sum_{i=1}^{t} w_i \svec_i$ of the (possibly differing Hamming weight) solution vectors $\svec_1,\svec_2,\cdots,\svec_t$, 
then we can similarly lower bound the length $\nrm{\svec}$.
 
\begin{lemma}
	\label{lem:LowerboundMFMultiple2}
	Suppose $\svec_1,\svec_2,\cdots,\svec_t $ are $0/1$ vectors and $\svec_1$ has the minimum Hamming weight. Then every vector in their convex hull $A$ has length ($\ell_2$-norm) at least $\nrm{\svec_1}/\sqrt{t}$.
\end{lemma}

\begin{proof}
	Let $\svec =\sum_{i=1}^{t} w_i \svec_i$ be an arbitrary vector in the convex hull $A$ generated by $\svec_1,\svec_2,\cdots,\svec_t $, where  $\sum_{i=1}^{t} w_i = 1$ and $w_i \geq 0$. Then
	$
	\nrm{\svec}^2 = \sum_{i=1}^{t}\sum_{j=1}^{t} w_i w_j \ip{\svec_i}{\svec_j}\geq \sum_{i=1}^{t} w_i^2 \ip{\svec_i}{\svec_i} \geq \sum_{i=1}^{t} w_i^2 \ip{\svec_1}{\svec_1} \geq \ip{\svec_1}{\svec_1}/t.
	$ 
	The first equality is by the definition of $\nrm{\svec}^2=\ip{\svec}{\svec}$. The first
	inequality is because $w_i w_j\ip{\svec_i}{\svec_j}\geq 0$ for any pair $ i,j\in[t]$. The second
	inequality is true because $\ip{\svec_i}{\svec_i} \geq \ip{\svec_1}{\svec_1}$ for any $i\in[t]$ as $\svec_i$ are a $0/1$ vectors and $\svec_1$ has the minimum Hamming weight. The third inequality follows from Cauchy-Schwarz.
\end{proof}
 
Combining \eqref{eq:condStrat} with \eqref{eq:vecMax}-\eqref{eq:vecTot}, \Cref{lem:LowerboundMFMultiple1} and \Cref{lem:LowerboundMFMultiple2}, we get our first lower bound result.

%{\color{blue} Mention the degree d in the 1st paragraph of the theorem? }
\begin{theorem}\label{thm:LowerboundMFMultiple}
	Suppose $\s_1,\s_2,\cdots,\s_t \in \{0,1\}^n$ are the $t$ solutions of
	$\calF=~{\calF}_1 \cup~\calF_2$, where $\calF_1=\{{ f_1}, \dots, { f}_m\} $  and
	$\calF_2=\{x_1^2-x_1,\ldots, x_n^2-x_n\}$, and let $h$ be the minimum Hamming weight of the $t$ solutions $\s_1,\s_2,\cdots,\s_t$.  Let $d$ be the selected degree on constructing the Macaulay linear system $\lemacmat \svec = \bvec$ and let $\svec_1,\svec_2,\cdots,\svec_t$ be the corresponding solution vectors of the Macaulay linear system 	$\lemacmat \svec = \bvec$ under the assignments $\s_1,\s_2,\cdots,\s_t$ respectively. 
	
	If all the $t$ solutions $\s_1,\s_2,\cdots,\s_t$ have the same Hamming weight $h$ or the minimum $\ell_2$-norm solution vector $\svec = \lemacmat^{+}\bvec$ is in the convex hull of $\svec_1,\svec_2,\cdots,\svec_t$, then the \tqlscn of $\lemacmat$ of $\calF$  in the Macaulay linear system   
 	\begin{itemize}
		\item using max degree is $\kappa_{\bvec}(\lemacmat) \geq \sqrt{\left((d+1)^{h}-1\right)/t}$, and
		\item using total degree is $\kappa_{\bvec}(\lemacmat) \geq \sqrt{\left(\binom{d+h}{h}-1\right)/t}$.
	\end{itemize}

 In particular in the setup in \cite{ChenGao2017}, using max degree $d= 3n$, we have $\kappa_{\bvec}(\lemacmat) \geq \sqrt{(3n)^{h}/t} $.  
	\end{theorem}

Now we give a lower bound in terms of the smallest Hamming weight in the binary solution set. For this we use the following purely geometrical lemma.

\begin{lemma}[Shortest vector within an affine subspace]\label{lem:shortestVec}
	Let $V\in\mathbb{C}^{n\times k}$ be a matrix with columns $v_1,v_2,\ldots, v_k$, and let $A$ be their (complex) affine hull $A:=\{Vx: x\in\mathbb{C}^k \text{ s.t. }\ip{x}{\onevec}=1\}$. Then $A$ contains the origin iff the column space of the Gram matrix $G=V^\dagger V$ does not contain $\onevec$. Moreover, the length-square $\dualvalue^*$ of the shortest vector (with respect to the $\ell_2$-norm) in $A$ is
	\begin{equation}\label{eq:GramDual}
	\dualvalue^*=\max\{\dualvalue \colon G-\gamma \onevec\cdot\onevec^T\succeq 0\}.
	\end{equation}	
	Furthermore, if $A$ does not contain the origin, then $\dualvalue^*=1/\ip{\onevec}{G^+\onevec}$ and the shortest vector in $A$ is $V\cdot w$ for $w = G^{+} \onevec / \ip{\onevec}{G^+\onevec}$.
\end{lemma}
Note that the problem of finding the length of the shortest vector in an affine subspace can also be reformulated as the following SDP:
\begin{equation*}%\label{eq:l2SDP}
\dualvalue^*=\min_{\rho \succeq 0} \tr(G \rho)
\text{  \qquad subject to: }   \tr(\onevec\cdot\onevec^T\rho) = 1,
\end{equation*}
where without loss of generality we can assume that the above optimizer $\rho$ has rank $1$. 

By the weak duality of SDPs, and utilizing the dual of the above problem, we get: 
\begin{equation*}%\label{eq:l2SDPduality}
\dualvalue^*\geq\max_{\gamma\in \mathbb{R}} \dualvalue
\text{  \qquad subject to: } G- \dualvalue \onevec\cdot\onevec^T \succeq 0.
\end{equation*}
In fact the following proof of \cref{lem:shortestVec} shows that the above inequality is tight, i.e., strong duality always holds for this SDP.
\begin{proof}
	The shortest vector $s$ in the affine subspace is orthogonal to all vectors of the form $v_i-v_j$, therefore we must have that $\ip{s}{v_i}$ is constant for all $i\in [k]$, i.e., $V^\dagger s \propto \onevec$. Since $s$ is in the column space of $V$, we can write it in the form $s=V\cdot w$ so we get $V^\dagger s=V^\dagger V\cdot w = G w \propto \onevec$. If $s\neq 0$ then this implies that $\onevec$ is in the column space of $G$, consequently $0\neq\ip{\onevec}{G^+ \onevec}$ and thereby $s=V\cdot w$ for $w=G^+ \onevec/ \ip{\onevec}{G^+\onevec}\!$.%
	\footnote{For the last implication note that we already showed $w=\beta G^{+} \onevec + v$, where $v\in\ker(G)$. As $\ker(G)=\ker(V)$ we can assume without loss of generality that $v=0$. It is easy to see that $VG^{+} \onevec / \ip{\onevec}{G^+\onevec}\in A$, and since $A$ is an affine subspace not containing the origin, it can only contain one vector of the form $\beta V G^{+} \onevec\colon \beta \in \mathbb{C}$, thus $s=VG^{+} \onevec / \ip{\onevec}{G^+\onevec}$. (Indeed, if two distinct vectors $x, y$ are in $A$ and $y=\lambda x$, then $ \frac{1}{1-\lambda}y-\frac{\lambda}{1-\lambda} x=0$ is also in $A$.)} 
	From this we can conclude $\dualvalue^*=\nrm{s}^2=\ip{w}{Gw}=1/\ip{\onevec}{G^+\onevec}$.
	
	Conversely, if $\onevec$ is in the column space of $G$ then $s=V\cdot w$ for $w=G^+ \onevec/ \ip{\onevec}{G^+\onevec}\!$ is a non-zero vector, which is the shortest vector in $A$ due to the fact that it is orthogonal to all vectors of the form $v_i-v_j$. We can conclude that $\dualvalue^*\neq 0$ iff $\onevec$ is in the column space of $G$.
	
	If $\dualvalue^*> 0$, we can formulate a ``dual'' optimization problem
	the following way: $\dualvalue^* = \max\{\gamma \colon 1-\gamma/\dualvalue^* \geq 0\} = \max\{\gamma \colon 1-\gamma\ip{\pmb{1}}{G^{+}\pmb{1}} \geq 0\} = \max\{\gamma \colon I-\gamma\sqrt{G^{+}}\onevec\cdot\onevec^T\sqrt{G^{+}} \succeq 0\}$. By multiplying the matrices with $\sqrt{G}$ from both sides we get the following equivalent maximization formulation
	\begin{equation}\label{eq:GramDualInvolved}
	\dualvalue^* =\max\{\gamma \colon G-\gamma \sqrt{G}\sqrt{G^{+}}\onevec\cdot\onevec^T\sqrt{G^{+}}\sqrt{G}\succeq 0\}.
	\end{equation}
	Note that $\sqrt{G^+}\sqrt{G}=\sqrt{G}\sqrt{G^+}=(G^+ G)$ is the orthogonal projector to the column space of $G$. Since $\onevec$ is in the image of $G$ the above equation \eqref{eq:GramDualInvolved} is equivalent to \eqref{eq:GramDual}.
	
	On the other hand, if $\dualvalue^* = 0$, then $\onevec$ is not in the column space of $G$ and so $(I-G^+ G)\onevec\neq 0$. Observe that $G-\gamma \onevec\cdot\onevec^T\succeq 0$ implies $(I-G^+ G)G(I-G^+ G)-\gamma (I-G^+ G)\onevec\cdot\onevec^T(I-G^+ G)\succeq 0$ or equivalently $-\gamma (I-G^+ G)\onevec\cdot\onevec^T(I-G^+ G)\succeq 0$, which then only holds for $\gamma\leq 0=\dualvalue^*$. Consequently, \eqref{eq:GramDual} holds even in the case $\dualvalue^*=0$.
\end{proof}
 
Suppose that the Boolean solution set of the polynomial system is 
$ S=\{a_1, a_2, \cdots, a_t\}$. Let $A_S$ be the affine subspace
corresponding to the solution set $S$, spanned by the monomial solution vectors $y_1, y_2, \cdots, y_t$ of the linear system $\lemacmat \svec=\bvec$ corresponding to the Boolean solutions $a_1, a_2, \cdots, a_t$ respectively.
We wish to lower bound the length of the shortest vector in $A_S$; for this it suffices to find the length of the shortest vector in an enlarged affine subspace $A_{S'} \supseteq A_S$.  

Let $S'$ be the symmetrized solution set of $S$ by applying all possible permutations of the variables of $a_i$'s and taking their union. Let $A_{S'}$ be the affine subspace corresponding to the solution set $S'$ and let $v$ be 
the shortest vector in $A_{S'}$.
For each Hamming weight $h$ that appears in $S'$, there is a symmetrized monomial solution vector $\textbf{v}_h$. This $\textbf{v}_h$ equals
the average over all monomial solution vectors that are associated to Boolean solutions of Hamming weight $h$. 
 
Next we will argue that the minimum $\ell_2$-norm vector $v \in A_{S'}$ is an affine combination of the symmetrized monomial solution vectors $\textbf{v}_h$.  If we apply an induced permutation on the coordinates of $v$ according to a permutation of the Boolean variables, then the $\ell_2$-norm of the resulting vector $u \in A_{S'}$ is equal to the $\ell_2$-norm of  $v$. 
Because $v$ has the minimum $\ell_2$-norm in $A_{S'}$ we have $ \nrm{\frac{u+v}{2}}_2 \geq \nrm{v} $, and due to $ \nrm{u} =\nrm{v}$ by the triangle inequality $ \nrm{\frac{u+v}{2}}_2 \leq \nrm{v} $ (equality holds if and only if $u=v$),  the resulting vector $u$ is equal to $v$. Therefore, the shortest vector $v$ is invariant under all the possible induced permutations, therefore we can conclude that $v$ is an affine combination of the symmetrized monomial solution vectors $\textbf{v}_h$.

Now, we can lower bound $\nrm{\lemacmat_{\calF}^+ b}$ by finding the lowest $\ell_2$ norm of a vector in the affine subspace spanned by the symmetrized vectors $\textbf{v}_h$ corresponding to the Hamming weights $h$ that appear in $S$. This can be achieved by considering the Gram matrix as explained in \Cref{lem:shortestVec}.

In order to compute this Gram matrix, we need to understand the symmetrized vectors $\textbf{v}_h$. For this, let us introduce the following orthonormal vector system $(\textbf{b}_s)$, corresponding to the set of monomials $\mathfrak{m}^d_s$ that contain exactly $s$ variables with a non-zero exponent, with the degree of the monomials being at most $d$, then $\textbf{b}_s:=\frac{1}{\sqrt{|\mathfrak{m}^d_s|}}\sum_{m\in\mathfrak{m}^d_s} e_m$. 
Also let $\Pi_{\mathfrak{m}^d_s}=\sum_{m\in\mathfrak{m}^d_s} e_m \cdot e_m^T$ be the projector to coordinates in $\mathfrak{m}^d_s$. Finally, let $c^d_s$ be the number of monomials that contain $s$ specific variables with a non-zero exponent and have degree at most $d$, so that $|\mathfrak{m}^d_s|=\binom{n}{s}c^d_s$.

One can see that for any $a\in\{0,1\}^n$ of Hamming weight $h$ we have $c^d_s\binom{h}{s}=\ip{\onevec}{\Pi_{\mathfrak{m}^d_s}\svec^{(a)}}=\ip{\onevec}{\Pi_{\mathfrak{m}^d_s}\textbf{v}_h}$. Since $\textbf{v}_h$ has uniform coordinates over $\mathfrak{m}^d_s$ we have 
$
\nrm{\Pi_{\mathfrak{m}^d_s} \textbf{v}_h}^2
=|\mathfrak{m}^d_s|
\left(\frac{\ip{\onevec}{\Pi_{\mathfrak{m}^d_s}\textbf{v}_h}}{\|\mathfrak{m}^d_s\|}\right)^{\!\!\!2}
=c^d_s\binom{n}{s}\left(\frac{c^d_s\binom{h}{s}}{c^d_s\binom{n}{s}}\right)^{\!\!2}
=c^d_s\binom{h}{s}^{\!\!2}/\binom{n}{s},
$
consequently
\begin{equation}\label{eq:symmeticVectors}
\textbf{v}_h=\sum_{s=1}^h \sqrt{c^d_s \binom{h}{s}^{\!\!2}/\binom{n}{s}}\textbf{\textbf{b}}_s,
\end{equation}
and the Gram matrix $G$ of the symmetrized vectors has matrix elements
\begin{equation*}%\label{eq:symmeticVectorsGram}
G_{ij}=\ip{\textbf{v}_i}{\textbf{v}_j}=\sum_{s=1}^n c^d_s \binom{i}{s}\binom{j}{s}/\binom{n}{s}.
\end{equation*} 

Together with \Cref{lem:shortestVec} this enables us to give a lower bound on the smallest $\ell_2$-norm solution in terms of the minimal Hamming weight appearing in the solution set as follows.

\begin{theorem}
	Suppose that $\calF$ is a Boolean polynomial system with $n$ Boolean variables where each solution have Hamming weight at least $h$, and $d\geq n$.  Recall that $\tqlscn$, defined as  $\kappa_{\bvec}(\lemacmat) =\nrm{\lemacmat}\frac{\nrm{\lemacmat^+\bvec}}{\nrm{\bvec}}$, is also lower bounded by the smallest $\ell_2$-norm solution by equation~\eqref{eq:condStrat}.
 Then the degree-$d$ Macaulay linear system's \tqlscn is lower bounded by
	\begin{equation}\label{eq:GramThm}
	\frac1{\ip{\onevec}{(G^{(h)})^{-1}\onevec}} = \max\{\dualvalue \colon G^{(h)}-\gamma \onevec\cdot\onevec^T\succeq 0\},
	\end{equation}
	where $G\in\mathbb{R}^{n\times n}$ is the Gram matrix whose $(i,j)$ matrix element is 
	\begin{equation}\label{eq:symmeticVectorsGram}
	G_{ij}=\sum_{s=1}^n c^d_s \binom{i}{s}\binom{j}{s}/\binom{n}{s},
	\end{equation} 
	 $c^d_s=d^s$ for max degree, while $c^d_s=\binom{d}{s}$ for total degree, and finally $G^{(h)}$ is the bottom-right $(n-h+1)$ by $(n-h+1)$ minor of $G$.\footnote{When $d\geq n$, due to the triangular shape of the non-zero coefficients of the vectors in \eqref{eq:symmeticVectors} it is easy to see that $G$ has full rank, i.e., it is positive definite. It follows, that all principal submatrix of $G$ are also positive definite, i.e., have full rank, therefore $G^{(h)}$ is invertible.} 
\end{theorem}

\jnote{By computing the same quantity two different ways we get 
	$$G_{h h}=\ip{\textbf{v}_h}{\textbf{v}_h}=\sum_{s=1}^h d^s \binom{h}{s}\binom{h}{s}/\binom{n}{s} \overset{!}{=} \left(\sum_{i=0}^{h-1} ((d+1)^{h-i}-1) \binom{h}{i}\binom{n-h}{i}\right)/\binom{n}{h}.$$}
\anote{Maybe we could use Stirling's approximation for replacing the binomial expressions in \eqref{eq:symmeticVectorsGram} with exponentials in order to prove a lower bound analytically.}

The expression in \eqref{eq:GramThm} is difficult to bound analytically, but it appears to be exponentially large in terms of $h$ for large enough $d$. In particular, we could verify\footnote{Aided by symbolic computations executed by Mathematica 12.3 on Linux. See
~\cite{SrcWolframNotebook2022}
 for the code.} that for max degree $d=3n$ \Cref{eq:GramThm} is lower bounded by $h^h/2$ for every $h\in[n]$ up to $n=300$.

\subsection{Comparison to brute-force search}\label{subsec:comparison}

In case there is a unique Boolean solution we showed that the lower bound of the running time of the quantum algorithm using the HHL algorithm is exponential in the Hamming weight of the unique Boolean solution, and we provided strong evidence that this is also true when there are multiple solutions. 

It is useful to compare the HHL-based approach to classical brute-force search and also to using Grover's algorithm.  
In case we know that the unique solution has Hamming weight $h$, we can simply classically search through all the $\binom{n}{h}$ different Hamming-weight-$h$ assignments of the original polynomial system. We can also use Grover search to find such an assignment with $\bigO{\sqrt{\binom{n}{h}}}$ evaluations of the polynomials. Even if we do not a priory know the Hamming weight, we can classically iterate over increasing Hamming weights $w$ of $n$-bit strings which requires at most $\bigO{\sum_{w=0}^{h}\binom{n}{w}}$ different possible assignments to be checked before finding the solution, which in the case $h\leq n/2$ can be bounded by $\bigO{\sqrt{h}\binom{n}{h}}$ as we show in \cref{apx:binom}. For the $h> n/2$ case, a similar complexity can be achieved by searching through decreasing Hamming weights. 
In the quantum case, naively iterating through increasing Hamming weights and using Grover's algorithm for each weight gives a complexity bound of $\bigO{h\sqrt{\binom{n}{h}}}$. 

Moreover, we can use a slight variant of Grover's algorithm for searching through an unknown sized search space\footnote{One can use an algorithm analogous to the ``exponential Grover search''~\cite{boyer1998TightBoundsOnQuantumSearching} in order to check for a unique solution in subsequently enlarged search spaces corresponding to larger and larger Hamming-weights. By carefully choosing the sequence of upper bounds on the Hamming weights such that the search space expands in each consecutive iteration by a bounded multiplicative factor in $[c,C]\subset (1,\infty)$ the claimed running time bound follows. } which requires only $\bigO{\sqrt[4]{h}\sqrt{\binom{n}{h}}}$ evaluations of the polynomials. By comparing this to the lower bounds of \Cref{thm:LowerboundMFMultiple} one can see that in case $d+h\geq n$ Grover's algorithm performs at least as good as the HHL based algorithm (up to some potential lower order correction $\sqrt[4]{h}$), and the algorithm of Chen and Gao~\cite{ChenGao2017} where the max degree $d=3n$ is definitely outperformed by Grover search.

In case there are multiple solutions, but all their Hamming weights are the same, \Cref{thm:LowerboundMFMultiple} ensures that we do not get a bigger reduction in the condition number than the analogous speedup we can already achieve by plain Grover search. So the above Grover-based algorithm still performs just as competitively. 

In the general case of having multiple solutions with different Hamming weights, the situation is harder to analyze, but we could still obtain an exponential lower bound on \tqlscn in terms of the smallest Hamming weight solution up to $n=300$, providing strong evidence for Chen and Gao's algorithm~\cite{ChenGao2017} having a best case complexity that is exponentially large in terms of the minimal Hamming weight of a solution, making it unlikely that their algorithm would give a substantial improvement over brute-force Grover search.
\section{The Boolean Macaulay linear system and its \tqlscn}

In this section we give an equivalent but more efficient way to
represent the Macaulay matrix using the fact that we are only
searching for 0/1 solutions in $\mathbb C$. This results in a smaller
lower bound on the \tqlscn of size $\Omega(2^{h/2})$. While the
quantum algorithm's running time is still exponentially large for larger Hamming weight
solutions, for Hamming weight $h=\Theta(\log n)$ the smaller lower bound leaves open the possibility of
a quasipolynomial speedup compared to the classical brute-force search algorithm having running time
$\bigO{\binom{n}{h}}$.

\subsection{The Boolean Macaulay matrix over \texorpdfstring{$\C$}{C}}

In this section we again include the field polynomials
$\calF_2 = \{x_1^2-x_1, \ldots, x_n^2-x_n\}$ for the field $\FF_2$ 
together with the input polynomials $\calF_1$. Solving the system
$\calF = \calF_1 \cup \calF_2$ forces the roots to be effectively Boolean even
though the underlying field is $\C$.  This allows all monomials in an
equation to be replaced with equivalent multilinear versions and a
reduced Macaulay matrix will be defined that has a more compact form.
This was done in~\cite{bardet2013complexity} for finite fields, where
the extra equations prevented solutions from being in field
extensions. 
We derive the analogous matrix when the solutions are forced to be Boolean but the arithmetic is over $\C$.
Additionally, in our case (similarly to Chen and Gao's original construction~\cite{ChenGao2017}) the structure of the Boolean solutions makes it possible to extract the Boolean solutions from measuring the quantum state corresponding the the solution vectors (over $\C$).

Let $\psi : R \rightarrow R$ map a monomial to its multilinear image as
$\psi(\prod_{i=1}^{n} x_i^{a_i} ) = \prod_{i=1}^{n} x_i^{\min
  \{a_i,1\} } $, and extend it to $R=\mathbb{C}[x_1,\dots, x_n]$ by
linearity.  For example,
$\psi(3x_1^3x_2 -1) = 3x_1x_2- 1 = \psi(x_1^3x_2 - 2 x_1^2 x_2^2 + 4
x_1 x_2 - 1) \neq x_1 - x_1x_2-1 = \psi(x_1^3 - 2 x_1^2 x_2^2 + x_1
x_2 - 1)$.

Lemma~\ref{lem:BooM} will show that having max degree
higher than $1$ becomes redundant, so the following definition only
has rows up to max degree $1$, and in this section we will set the total 
degree $d=n$, the number of variables.  For notation, let $\mon$ and $\monb$ denote
monomials, and let $\mlmon$, $\mlmonb$, and $\mlmonc$ denote
multilinear monomials (i.e., monomials with max degree at most 1).

\begin{definition}\label{def: BooleanMacaulay}
  The {\em Boolean Macaulay matrix $\bmacmat$ of degree $d$} of
  $\calF_1=\{f_1,\ldots,f_m\}\subseteq \CC[X]$ is the matrix where each
  row is labeled by a pair of polynomials $(\mlmon,f)$ and contains
  the corresponding coefficient vector of the polynomial $\psi(\mlmon f)$.
  The rows range over all $f\in \calF_1$ and multilinear monomials $\mlmon$ such
  that $\psi(\mlmon f)$ has degree at most $d$.  The columns are labeled by
  the set of multilinear monomials in $x_1,\ldots, x_n$ of degree at most $d$ and
  are ordered with respect to a specified monomial ordering.  The matrix
  element in the row corresponding to $(\mlmon,f)$ and the column
  corresponding to the monomial $\mlmonb$ is the coefficient of $\mlmonb$
  in the polynomial $\psi(\mlmon f)$.
\end{definition}

Note that compared to the Macaulay matrix, in addition to forcing
answers to be Boolean, the Boolean Macaulay matrix is reduced in a
certain way, by eliminating polynomials with max degree at least 2.
Next we will show that the Boolean Macaulay matrix can be obtained as
a submatrix of the Macaulay matrix $\macmat $ of max degree $d$
corresponding to the set of polynomials $\calF = \calF_1 \cup \calF_2$
after Gaussian reduction on the rows.

First we consider the special case when $\calF_1 = \emptyset$ and
perform the row reduction on the Macaulay matrix of $\calF_2$ to show
that the field equations $\calF_2$ take a special form.

\begin{lemma}
  \label{lem:BooMempty}
  Let the Macaulay matrix $\macmat_2$ of max degree $d$ of $\calF_2$ have its columns ordered such that they are partitioned into two parts the following way: 
  the labels on the right side are  multilinear monomials (including the degree 0 monomial 1) and ordered in ascending order with respect to the integer represented by the exponent vector of the multilinear monomial, 
  and let the left side columns be labeled by nonmultilinear monomials and ordered under any monomial order.

  Then using row operations, $\macmat_2 = [L_2\ \ R_2]$ can be reduced
  to $\macmat_2'= [ I_2 \ \   B_2 ]$ where $ I_2$ is the
  identity matrix of dimension $(d+1) ^ n-2^n$ with rows and columns
  labeled by nonmultilinear monomials, and rows with zeros are removed.
\end{lemma}

\begin{proof}
  The rows of $\macmat_2$ are indexed by pairs of polynomials
  $(\mon,x_j^2-x_j)$, where $\mon$ is a monomial and
  $\maxdeg \mon (x_j^2-x_j) \leq d$.  The approach is to first change
  each row, which starts with coefficients for a polynomial
  $\Pi_i x_i^{a_i} - x_j^{-1}\Pi_i x_i^{a_i}$ for some $j$ in
  $\macmat_2$, to the coefficients of
  $\Pi_i x_i^{a_i} - \Pi_i x_i^{\min \{a_i,1\}}$.  At this point the
  left side of the matrix has at most one 1 in each row.  The second
  step is to zero out the bottom rows.

% {\color{blue} ``the rows where prod is the highest degree term"}
  For the first step, work in descending total degree of the
  polynomials, starting at degree $nd$.  Let the current row have the
  coefficients of $\Pi_i x_i^{a_i} - \Pi_i x_i^{b_i}$ during the
  algorithm.  Let $b_j \geq 2$ for some $j$, or else this row is
  reduced.  Because $\deg \Pi_i x_i^{b_i} < \deg \Pi_i x_i^{a_i}$, the rows where $\Pi_i x_i^{b_i}$ is the highest degree term have not changed yet, and therefore, one
  of the rows has the coefficients of
  $\Pi x_i^{b_i} - x_j^{-1} \Pi x_i^{b_i}$.  Adding this row changes
  the current row to
  $\Pi_i x_i^{a_i} - \Pi_i x_i^{b_i} +(\Pi_i x_i^{b_i} - x_j^{-1}
  \Pi_i x_i^{b_i}) = \Pi_i x_i^{a_i} - x_j^{-1} \Pi_i x_i^{b_i}$,
  which has decreased the total degree of the second term by one while
  keeping the set of variables the same.  This is repeated until the row has the
  coefficients of $\Pi_i x_i^{a_i} - \Pi_i x_i^{\min \{a_i,1\}}$.

  At the end of the first step, each row in the left side (i.e.,
  columns indexed by nonmultilinear monomials) has at most one 1.
  This is in fact a constructive argument showing that 
  \[
    \prod_{i=1}^{n} x_i^{a_i} - \prod_{i=1}^{n} x_i^{\min \{a_i,1\} }
    \in \left\langle \calF_{2}
    \right\rangle.
  \]
	
  Consider any two rows indexed by $\mon (x_i^2-x_i)$ and
  $\monb(x_j^2- x_j)$.  If $\mon x_i^2 = \monb x_j^2$, they have the
  same set of variables, and therefore
  $\psi(\mon x_i) = \psi(\monb x_j)$, so the rows are equal and one
  can be eliminated (zeroed out). Keep doing this until for every leading nonmultilinear
  monomial there is only one row where the corresponding coefficient is nonzero.
	
	Because for every column in the left part there is a unique nonzero row with the corresponding leading monomial, the matrix can be written (up to permutation of the rows) as
  $\begin{bmatrix} I_2 & B_2\\
    0 & 0 
  \end{bmatrix}$.
\end{proof}

\begin{lemma}  \label{lem:BooM}
  Let $
  \macmat=
  \begin{bmatrix}
    L_1 & R_1\\
    L_2 &  R_2
  \end{bmatrix}
  $ be the Macaulay matrix for $\calF_1 \cup \calF_2$ with the row and
  column ordering from Lemma~\ref{lem:BooMempty}.  Using row
  operations (and then removing some zero rows), $\macmat$ can be reduced to $ \macmat'=
	\begin{bmatrix}
	0 & \bmacmat\\
	I_2 &  B_2
	\end{bmatrix}
	$ where $\bmacmat$ is the Boolean Macaulay matrix of $\calF_1$, 
	and $I_2, B_2$ are as in Lemma~\ref{lem:BooMempty}. 
\end{lemma}

\begin{proof}
  By Lemma~\ref{lem:BooMempty}, using row operations on the $\calF_2$ submatrix of 
  $\macmat$ we get a matrix 
  $\begin{bmatrix}
    L_1 & R_1 \\
    I_2 &  B_2\\
  \end{bmatrix}$, where zero rows from $\calF_2$ are removed.
	
  Row operations utilizing $I_2$ can then be used to zero out the top left, resulting in 
  $\begin{bmatrix}
    0 & R_1'\\
    I_2 &  B_2
  \end{bmatrix}$.
  From the polynomial perspective, this maps all the nonmultilinear polynomials to their
  corresponding multilinear polynomials under $\psi$, 
  i.e., for each monomial $\prod_{i=1}^{n} x_i^{a_i}$, the map encodes the coefficient vector of
  $\psi((\prod_{i=1}^{n} x_i^{a_i}) f_j), 1 \leq j \leq m$ into the Macaulay matrix as a row vector.

Recall that the Macaulay matrix has rows labeled by pairs; observe that rows $(\mon,f_i)$ and $(\monb,f_i)$ will be equal when $\psi(\mon) =\psi(\monb)$. In particular for any nonmultiliear monomial $\mon$ the rows $(\mon,f_i)$ and $(\psi(\mon),f_i)$ will be equal at this point, so we can eliminate (zero out and then remove) any row indexed by nonmultiliear monomials. To be compatible with \Cref{def: BooleanMacaulay} we choose not to further reduce / remove rows despite the fact $\bmacmat$ might have zero rows, for example, rows with $\psi(\mon f)=\psi(\monb f')$, but $f\neq f'$.

  As claimed, in matrix notation we get $ \macmat'=
  \begin{bmatrix}
    0 & \bmacmat\\
    I_2 &  B_2
  \end{bmatrix}. $ 
\end{proof}

As in the general case let $\bmacmat = [M \quad -\bvec]$ define the
Boolean Macaulay linear system as $M \svec = \bvec$,
where the entries of $\svec$ are labeled by the nontrivial multilinear
monomials and $\bvec= \begin{bmatrix}
  1\\
  \textbf{0}
\end{bmatrix}$.  

Recall that a matrix is $s$-sparse if it has at most $s$ entries in
any row or column. 

\begin{lemma}\label{lemma:sparsity}
  The Boolean Macaulay matrix $\bmacmat$ of total degree $d$ of
  $\calF_1$ is an $\bigO{m \cdot \sparsity{\calF_1}}$-sparse matrix.
\end{lemma}
\begin{proof}
  The Boolean Macaulay matrix $\bmacmat$ is constructed by placing $\psi(\mlmond f)$ in a row for a multilinear monomial $\mlmond$ and  $f\in \calF_1$, so the support of each row has size at most  $\bigO{\sparsity{\calF_1}}$.

For the column sparsity first consider the Boolean Maculay matrix of $\{\mlmond\}$, which has a $1$ matrix element at column $\mlmonb$ and row $(\mlmonc,\mlmond)$ if and only if $\mlmonb=\psi(\mlmonc\cdot\mlmond)$. This can only happen if $\mlmond$ divides $\mlmonb$, so we can define $\bar{\mlmond}:=\mlmonb/\mlmond$. It is easy to see that $\mlmonb=\psi(\mlmonc\cdot\mlmond)$ if and only if $\mlmonc = \bar{\mlmond}\cdot\mlmon_d$ for some monomial $\mlmon_d$ that divides $\mlmond$. This implies that the column sparsity of the Boolean Maculay matrix of $\{\mlmond \}$ equals the number of divisors of $\mlmond $ which is at most $4$ if the multilinear monomial $\mlmond$ has (total) degree at most $2$.

Now consider the Boolean Maculay matrix of $\{f\}$ for some (at most) quadratic polynomial $f \in \calF_1$. Observe that $f$ is a linear combination of at most $\sparsity{\calF_1}$ monomials of degree at most $2$ and the Boolean Maculay matrix of $\{f\}$ is likewise the linear combination of the Boolean Maculay matrix of these (at most) quadratic monomials. So the column sparsity of the Boolean Maculay matrix of $\{f\}$ is at most $4\cdot \sparsity{\calF_1}$. Finally, the entire Boolean Macaulay matrix of $\calF_1$
is simply given by stacking the Boolean Maculay matrices of $\{f\}$ for $f \in \calF_1$, 
so the total column sparsity is at most $4m\cdot \sparsity{\calF_1}$.
\end{proof}

Note that this also implies that $M$, a submatrix of $\bmacmat$, is also sparse. Moreover, the location and value of the nonzero entries of each column/row of $\bmacmat$ can be efficiently computed.

Now we show that the Boolean Macaulay linear system is equivalent to the Macaulay linear system. It follows that solving the Boolean Macaulay linear system returns a correct solution of the Boolean polynomial system.
%{\color{blue} Refer appendix A properly here}
\footnote{This observation also implies that the complete solving degree in Chen and Gao's original approach is always at most $n+2$, tightening their upper bound $3n$.}

\begin{lemma}\label{lem: Macaulaylinear}
	
	Let $M_1 \vec{y}_1 =\vec{b}_1$ be the Macaulay linear system of a polynomial system $\calF=\calF_{1} \cup \calF_{2}$ and let $M_2 \vec{y}_2 = \vec{b}_2$ be the corresponding Boolean Macaulay linear system, where the Macaulay matrix is $\macmat = [M_1 \quad -\vec{b}_1]$, the Boolean Macaulay matrix is $ \bmacmat = [M_2 \quad -\vec{b}_2]$. Then a solution $\hat{y}_2$ of the Boolean Macaulay linear system
	$M_2 \vec{y}_2 = \vec{b}_2$  corresponds to a solution $\hat{y}_1$ of
	the Macaulay linear system $M_1 \vec{y}_1 =\vec{b}_1$.
\end{lemma}

\begin{proof}
	 By Lemma~\ref{lem:BooM}, 
	$\macmat = [M_1 \quad -\vec{b}_1]$
	can be reduced to  $	\begin{bmatrix}
	0 & \bmacmat\\
	I_2 &  B_2
	\end{bmatrix}$ by row operations, where  $\bmacmat = [M_2 \quad -\vec{b}_2]$ is the Boolean Macaulay matrix.
	 Also, $ B_2 = [B_2'  \quad 0]$ because the last column of the reduced Maculay matrix  $	\begin{bmatrix}
	0 & \bmacmat\\
	I_2 &  B_2
	\end{bmatrix}$ is indexed by the degree $0$ monomial $1$ and the polynomials generated from $\calF_{2}$ have no constant terms. Therefore $	\begin{bmatrix}
	0 & \bmacmat\\
	I_2 &  B_2
	\end{bmatrix}=
	\begin{bmatrix}
	0 & M_2 & -\vec{b}_2\\
	I_2 &  B_2'  &   0
	\end{bmatrix}
	$.
	
Since performing row operations on the augmented matrix of a linear system does not change the set of solutions,
	solving the Macaulay linear system 
	$M_1 \vec{y}_1 =\vec{b}_1$
	is equivalent to solving the linear system 
	$
	\begin{bmatrix}
	0 & M_2\\
	I_2 &  B_2'
	\end{bmatrix}
	\begin{bmatrix}
	\vec{z}_1 \\
	\vec{y}_2
	\end{bmatrix}=
	\begin{bmatrix}
	\vec{b}_2\\
	0
	\end{bmatrix}
	$, where the entries of $\vec{y}_2$ and $\vec{z}_1$  are indexed by nontrivial multilinear monomials and nonmultilinear monomials respectively. 

	For the linear system 
	\[ 
	\begin{bmatrix}
	0 & M_2\\
	I_2 &  B_2'
	\end{bmatrix}
	\begin{bmatrix}
	\vec{z}_1 \\
	\vec{y}_2
	\end{bmatrix}=
	\begin{bmatrix}
	\vec{b}_2\\
	0
	\end{bmatrix}	
 	\] 
	we have $M_2 \vec{y}_2 = \vec{b}_2$, which is the Boolean Macaulay linear system, and $\vec{z}_1 + B_2' \vec{y}_2 = 
	0$. 
	
 If $\hat{y}_2$ is a solution of the Boolean Macaulay linear system $M_2 \vec{y}_2 = \vec{b}_2$ , set $\hat{z}_1 $ to be $- B_2' \hat{y}_2 $, then 	$
	\begin{bmatrix}
	\hat{z}_1\\
	\hat{y}_2
	\end{bmatrix}$ is a solution of the linear system $
	\begin{bmatrix}
	0 & M_2\\
	I_2 &  B_2'
	\end{bmatrix}
	\begin{bmatrix}
	\vec{z}_1 \\
	\vec{y}_2
	\end{bmatrix}=
	\begin{bmatrix}
	\vec{b}_2\\
	0
	\end{bmatrix}
	$. Because the Macaulay linear system is equivalent to the linear system $\begin{bmatrix}
	0 & M_2\\
	I_2 &  B_2'
	\end{bmatrix}
	\begin{bmatrix}
	\vec{z}_1 \\
	\vec{y}_2
	\end{bmatrix}=
	\begin{bmatrix}
	\vec{b}_2\\
	0
	\end{bmatrix}$, therefore, a solution $\hat{y}_2$ of the Boolean Macaulay linear system $M_2 \vec{y}_2 = \vec{b}_2$ corresponds to a solution $\hat{y}_1$ of
	the Macaulay linear system $M_1 \vec{y}_1 =\vec{b}_1$.
\end{proof}

As $M$ is a $\bigO{m\cdot\sparsity{\calF}}$-sparse row / column computable
matrix and we can efficiently prepare the sparse vector $\bvec$ as quantum state
$\left| b\right\rangle$, we can apply a QLS algorithm to ``solve'' the Boolean
Macaulay linear system $ M \svec = \bvec$, which takes time
$\bigOt{\text{poly}(n)\kappa(M) \log (1/\epsilon)}$ \cite{childs2015QLinSysExpPrec}.

The key parameter in the running time is the condition number of the matrix
$M $. Next we will provide a lower bound of the \tqlscn of $M$ and thus also a lower bound on known QLS algorithms.  

\subsection{Lower bound on the \tqlscn  \texorpdfstring{$\kappa_{\bvec}(M)$}{k{\scriptsize b}(M)} }

Suppose $\s_1,\s_2,\cdots,\s_t \in \{0,1\}^n$ are the $t$ solutions of
$\calF=~{\calF}_1 \cup~\calF_2$, where $\calF_1=\{{ f_1}, \dots, { f}_m\} $  and
$\calF_2=\{x_1^2-x_1,\ldots, x_n^2-x_n\}$, and let $h$ be the minimum Hamming weight of the $t$ solutions $\s_1,\s_2,\cdots,\s_t$. Let $\svec_1,\svec_2,\cdots,\svec_t$ be the corresponding solution vectors of the Boolean Macaulay linear system 
$M \svec= \bvec$ under the assignments $\s_1,\s_2,\cdots,\s_t$ respectively. 

In this case, we have $\nrm{M}\geq 1/2$ as $M$ has at least one matrix element which has an absolute value at least $1/2$ \footnote{After applying \ref{it:redN} there is at least one polynomial $f$ with a constant term of magnitude $1$. If $f$ does not have a degree-$1$ monomial $x_i$, then $x_i\cdot f$ has a magnitude $1$ degree-$1$ monomial $x_i$, so the row $(x_i,f)$ will have a matrix element of magnitude $1$. Otherwise, suppose the coefficient of $x_1$ in $f$ is $c_1$, then the rows $(1,f)$ and $(x_1,f)$  will have a matrix element of magnitude $c_1$ and $c_1-1$ respectively. Therefore, at least one of them has a magnitude of at least $1/2$.}. Analogously to \Cref{thm:LowerboundMFMultiple} we get:

\begin{corollary}\label{cor:LowerboundBooleanMFMultiple}
Let $\bmacmat = [M \quad -\vec{b}]$ be the Boolean Macaulay matrix of $\calF$ with columns labeled by multilinear monomials. Let $h$ be the minimum Hamming weight of the $t$ solutions $\s_1,\s_2,\cdots,\s_t$.
If all the $t$ solutions $\s_1,\s_2,\cdots,\s_t$ have the same Hamming weight $h$ or the minimum $\ell_2$-norm solution vector $\svec = M^{+}\vec{b}$ is in the convex hull of $\svec_1,\svec_2,\cdots,\svec_t$,
then the \tqlscn $\kappa_{\bvec}(M) \geq \frac{1}{2}\sqrt{(2^h-1)/t}$  of $M$ of $\calF$.  
\end{corollary}

For $h= \Theta(\log n)$, this lower bound does not rule out the
possibility that the Macaulay matrix has a polynomial condition
number, which would result in the quantum algorithm beating the brute-force classical algorithm that runs in time $\bigOt{\binom{n}{\log n}}$.

\subsection{Details comparing running times}

As we have discussed in \Cref{subsec:comparison} the classical brute-force algorithm
tries all $\binom{n}{j}$ choices for the locations of the 1's in the
solution $\s$ for each $j\leq h$, and its running time can be bounded $ \bigO{ \sqrt{h}\binom{n}{h}}$, where
\[ \forall  1\leq h \leq n: \left( \frac{n}{h} \right)^h \leq \binom{n
}{h} \leq \left( \frac{en}{h} \right)^h.\]  
Comparing the above expression with our \tqlscn lower bound, 
we saw that the Gorver-enhanced brute-force search always outperforms 
the Maculay matrix approach in case there is a unique solution and $d=n$ (or even $d+h\geq n$).
This in particular shows that the quantum algorithm achieves at most a quadratic speed-up compared to classical brute-force search. Moreover, if one chooses $d=3n$ and works with the max degree as Chen and Gao suggested~\cite{ChenGao2017}, then $\kappa_{\bvec}(\lemacmat) \geq (3n)^{h/2}$ and so
\begin{itemize}
	\item For $h = \Omega(\sqrt{n})$, the classical brute
	force algorithm is faster than the quantum algorithm. 
	\item For $h = \bigO{\sqrt{n}}$, it is unknown which is faster.
\end{itemize}

On the other hand in the Boolean case we have only the lower bound $\kappa_{\bvec}(M) \geq \frac{1}{2}\sqrt{(2^h-1)}$, so:

\begin{itemize}
	\item For $h$= $pn$, where $p \in (0,\frac{1}{2}]$,   the lower bound of
	$\kappa_{\bvec} (M)\geq \frac{1}{2}(2^h-1)^{1/2}$ 
	is exponentially large and exhaustive search takes time $\bigO{2^{H(p) n}}$
	where $H(p)=-p\log p-(1-p)\log (1-p)$ is the binary entropy function, as shown in \Cref{apx:binom}.
	\item For $h$= $\bigO1$, $\exists$ classical algorithm that takes time 
	$\bigO{\binom{n}{h}}$ to solve the problem efficiently by
	exhaustive search whereas the 
	lower bound of $\kappa(M)$ is a constant $(2^{\bigO1}-1)^{1/2 }$.
	\item For $h$= $\Theta (\log n)$, we only know that $ \kappa_{\bvec}(M)\geq \text{poly}(n)$ whereas 
	classical exhaustive search takes time $\bigO{\binom{n}{\log n}}$. 
	Thus, we cannot exclude the possibility that the quantum algorithm might give a quasi-polynomial speedup in this case. 
\end{itemize}
Without loss of generality, let $0 \leq h \leq \frac{n}{2}$ (otherwise we can flip all variables). 
Then the lower bound on the \tqlscn $\kappa(M)$ is always smaller than the time required by brute-force search. 
Thus, there is a possibility that the quantum algorithm performs better than the exhaustive search approach.

\section{Our new improved quantum algorithm}

\subsection{A Variant of the Quantum Coupon Collector Problem}
  
By~\cite[Corollary
3.19]{ChenGao2017} and Lemma~\ref{lem: Macaulaylinear}, if a set of polynomials $\calF$ over $\CC[x_1,\ldots,x_n]$ has a unique solution
$\s = (\s_1, \s_2, \cdots, \s_n)\in \{0,1\}^n$, then for some $d$ less than or equal to $n$,
the corresponding Boolean Macaulay linear system $M\svec = \bvec$ of total degree $d$ has a 
unique solution $\svec = M^{+}\bvec$,
where the entries of  $\svec$ are indexed by multilinear monomials in
$x_1,\ldots,x_n$ with total degree at most $d$.
Let $U=\{x_1,x_2,\dots,x_n\}$. There is a one-to-one correspondence between the subsets of
$U$ of size at most $d$ and multilinear monomials in $x_1,\ldots,x_n$
with total degree at most $d$. 
Let $S$ be the largest subset of $U$ such that all the variables $x_k \in S$ have
assignment $\s_k = 1$ and $S_d$ be the set containing all nonempty subsets of $S$
that have size at most $d$. There is a one-to-one correspondence between the elements of the
 set $S_d$ and the nonzero entries of $\svec$. 
Given implicit access to matrix $M$ and sparse vector $b$, the QLS algorithm outputs the solution vector $\svec$ as the quantum state
$\left| \svec\right\rangle $, which encodes the nonzero entries of $\svec$. Because the unique solution $\svec = M^{+}\bvec$ of the Boolean Macaulay linear system is a 0/1 vector,
 the quantum state 
$\left| \svec\right\rangle $ can be represented by
\[
	\left| \svec \right\rangle =\frac{1}{\sqrt{|S_d|}} \sum_{R \in S_d} \left| R \right\rangle .
\]
If we measure the quantum state $\left|  \svec\right\rangle $, we will get a uniformly random subset
$R \in S_d$, where all the variables $x_k \in R$ have assignment $\s_k = 1$.
Given copies of the quantum state $\left| \svec\right\rangle$, the goal is to compute $S$.

Next, we will reformulate this problem as a variant of the quantum coupon collector problem.

\begin{problem}\label{prob:QVCCP} 
	Let $S \subseteq U=\{x_1,x_2,\dots,x_n\}$ be an unknown subset and  $S_d$ be the set containing all nonempty subsets of $S$ that have size at most $d$.
	Given copies of the state
	\[
	\left| \svec \right\rangle =\frac{1}{\sqrt{|S_d|}} \sum_{R \in S_d} \left| R \right\rangle,
	\] which is a superposition of subsets of $S$ of size at most $d$.
	The goal is to compute $S$.
\end{problem} 
Specially, when $d$ equals $1$, this is the quantum coupon collector problem
defined in~\cite{arunachalam2020QuantumCouponCollector}. They proved that $\Theta(|S|\log (\text{min}\{|S|, n-|S|\}))$ copies of the states $\left| \svec \right\rangle$ are necessary to compute $S$.

Without loss of generality, we can assume $d$ is at most $|S|$ because when $d$ is greater than $|S|$, the quantum state $\left| \svec \right\rangle$ is the same as the case of $d$ equals $|S|$. Then, we have the following result of Problem~\ref{prob:QVCCP}.

\begin{theorem} \label{thm: QVCCP} 
Let $r=\bigO{(|S|/d )\log (|S|/\eps)}$. Measuring $r$ copies of the quantum superposition state in Problem~\ref{prob:QVCCP}, the set $S$ can be computed with probability at least $1 - \eps$.
\end{theorem}
Since the only quantum operation is a measurement in the computational basis, this is essentially a classical coupon collector problem, where we can sample a uniformly random subset. 
\begin{proof} 	 
For any $x\in S$, the number of sets $R\in S_d$ containing $x$ is 
$\sum_{i=1}^{d} \binom{|S|-1}{i-1} $ out of a total number of sets 
$\sum_{i=1}^d \binom{|S|}{i}$ in $S_d$.  Thus, the probability of seeing
$x$ equals $ \sum_{i=1}^{d}
\binom{|S|-1}{i-1} /\sum_{i=1}^d \binom{|S|}{i} $.
If $0 \leq d \leq
\lfloor\frac{|S|}{3}\rfloor $ and by Appendix \ref{apx:binom}, we have 
$\binom{|S|}{d} \leq \sum_{i=1}^{d}\binom{|S|}{i} \leq \binom{|S|}{d}
\frac{|S|-d+1}{|S|-2d+1}$, so $ \sum_{i=1}^{d}
\binom{|S|-1}{i-1} /\sum_{i=1}^d \binom{|S|}{i} \geq \frac{d}{|S|}\cdot\frac{|S|-2d+1}{|S|-d+1} $, where $\frac{|S|-2d+1}{|S|-d+1}>\frac{1}{2}$.  Hence,  when $0 \leq d \leq
\lfloor\frac{|S|}{3}\rfloor $, the probability of not seeing $x$ after $r$ tries is at most $(1-\frac{d}{|S|} \cdot \frac{1}{2})^r = (1-\frac{d}{2|S|} )^{- \frac{2|S|}{d} \frac{d}{2|S|}  r}  \leq  \exp(-\frac{d}{2|S|} r) $.   Since
$\frac{\binom{|S|-1}{0} }{\binom{|S|}{1}} <  \frac{\binom{|S|-1}{1}
  }{\binom{|S|}{2}} < \dots <  \frac{\binom{|S|-1}{d-1} }{\binom{|S|}{d}} $, the probability function $ \sum_{i=1}^{d}
\binom{|S|-1}{i-1} /\sum_{i=1}^d \binom{|S|}{i} $ is an increasing function. If $\lfloor\frac{|S|}{3}\rfloor \leq d \leq  |S|$ and $|S|=\Omega(1)$ \footnote{Note that when $|S|=O(1)$, the probability of seeing $x$ is at least a constant.},
 the probability function has the minimum value when $d = \lfloor\frac{|S|}{3}\rfloor $. For all three cases, we have
 $$
\frac{d}{|S|}\cdot\frac{|S|-2d+1}{|S|-d+1} = \left\{
 \begin{array}{cc}
 \frac{d+1}{6d+3} & \text{ when } |S|=3d \\
 
 \frac{d^2+2d}{6d^2+8d+2} & \text{ when } |S|=3d+1 \\
 
 \frac{d^2+3d}{6d^2+13d+6}& \text{ when } |S|=3d+2
 \end{array}
 \right.
 $$
 the probability of seeing $x$ is greater than  $ \frac{d}{|S|}\cdot\frac{|S|-2d+1}{|S|-d+1} \geq \frac{1}{6}$. Hence, the probability of not finding $x$ after $r$ tries is at most $(1-\frac{1}{6})^r$.

Let $r=\bigO{(|S|/d )\log (|S|/\eps)}$, and by the union bound, the probability of not collecting all the
elements $x$ in $S$ is at most $\eps$. That is, if  $r=\bigO{(|S|/d )\log (|S|/\eps)}$, the entire set $S$ can be recovered with probability at least $1 - \eps$.
\end{proof}
With respect to the choice of $d$, there is a trade-off between the number of samples and the memory space: 
\begin{itemize}
	\item For $d=O(1)$,	 $r = O(|S|\log |S|)$.
	\item For $d=O (\log |S|)$,  $r = O(|S|)$.
	\item For $d=O(\frac{|S|}{\log |S|})$,  $r = O(\log^2 |S|)$ 
	\item For $\frac{|S|}{c} \leq d\leq  n$, where $c$ is a positive integer, $r=O(\log |S|)$.
\end{itemize}

%For polynomial systems over a finite field, the Hamming weight of the solution $\s $ of $\calF$ is $|S|=\Theta(n)$ when reformulated as Problem~\ref{prob:MQC}. 
%In this case, let $d = n$, then it takes $O(\log n)$ measurements to find the solution of $\calF$  from $ \left| \svec \right \rangle $ with probability at least $1- 1/\text{poly}(n)$.

\subsection{The algorithm}
When a set of polynomials $\calF$ has a unique solution, \Cref{alg:MQC} finds the solution. If a set of polynomials has more than one solution, we apply the Valiant-Vazirani reduction~\ref{it:redVV} to get a set of polynomials $\calF$ that have a unique solution. 

 \begin{algorithm}[H]
    \caption{ Quantum linear system algorithm for $\calF$
		over $\mathbb{C}$}
	\begin{algorithmic} \label{alg:MQC}
		\REQUIRE $\calF \subseteq \mathbb{C}[x_1, \dots,x_n]$ 
		where	$\calF=\{{ f_1}, \dots, { f}_m\}  $ with
		$\deg(f_i) = 2$ for $i=1,\dots, m$.
		\ENSURE The solution $\s \in \{0,1\}^n$ such
		that $ { f_1} (\s) = \cdots= { f}_m(\s) = 0$ over $\CC$ when one exists.		
		\STATE Step 1: Apply a quantum linear system algorithm to the Boolean Macaulay linear system  $M\svec = \bvec$ of total degree $n$ and get the solution $\svec$ in quantum state 	\[
	\left| \svec \right\rangle =\frac{1}{\sqrt{|S_d|}} \sum_{R \in S_d} \left| R \right\rangle
	\].
	\STATE Step 2:  Perform measurement on the quantum state $\left| \svec
	\right\rangle $ and get outcome $\left| R \right\rangle $, then let all the variables in the set $R$ equal 1.
	\STATE Step 3:  Repeat Step 1 and Step 2  $ O (\log n)$ times, and then set all left remaining variables $\s_j = 0$.
	\STATE Step 4: Return  $\s$.
	\end{algorithmic}
  \end{algorithm}

\begin{lemma}
	With high probability Algorithm~\ref{alg:MQC} solves Problem~\ref{prob:MQC} in
	time $\bigOt{ \mathrm{poly}(n)\kappa(M)}$.
\end{lemma}

\begin{proof}
For the Boolean Macaulay linear system $M\svec = \bvec$, the matrix $M$ is $\bigO{m \cdot \sparsity{\calF}}$-sparse and the vector $\vec{b}$ can be prepared as $| \bvec\rangle = \left| 0 \right\rangle^{ n \lceil{\log m}\rceil} $. Therefore, we can apply a QLS algorithm~\cite{childs2015QLinSysExpPrec} to the Boolean Macaulay linear system, which takes time
	$\bigOt{\mathrm{poly}(n) \kappa (M) \log (1/\eps)}$. The QLS algorithm outputs a quantum state $\left| 
	\svec^*\right\rangle $, which is an approximation of $\left| \svec\right\rangle$ with 
	$\nrm{\ket{\svec}  - \ket{\svec^*}} \leq \eps$.  
If we repeat the process $r$-times, then we essentially prepare the state $\ket{(\svec^*)^{\otimes r}}$ for which we have $\bracket{(\svec)^{\otimes r}}{(\svec^*)^{\otimes r}}=(\bracket{\svec}{\svec^*})^{r}=(1-\Theta(\eps^2))^r$. For $\eps=\bigO{1/r}$ we have that this equals $(1-\Theta(r\eps^2))$ and so $\nrm{\ket{(\svec)^{\otimes r}}  - \ket{(\svec^*)^{\otimes r}}} = \Theta(\sqrt{r}\eps)$.
Then the total variation distance between the two probability distributions of any measurements on the two states $\ket{(\svec)^{\otimes r}}$ and $\ket{(\svec^*)^{\otimes r}}$ is at most $\Theta(\sqrt{r}\eps)$~\cite[Exercise 4.3]{wolf2019QCLectureNotes} (see also \cite[Lemma 3.6]{bernstein1997quantum}), so replacing the ideal state $\ket{(\svec)^{\otimes r}}$ by the approximate state $\ket{(\svec^*)^{\otimes r}}$ induces error probability at most $\bigO{\sqrt{r}\eps}$.
By Lemma ~\ref{thm: QVCCP}, we can extract the solution of the polynomial system $\calF$ from $\ket{(\svec)^{\otimes r}}$ with high probability by choosing $r$ to be $O(\log n)$. 
Therefore, letting $\eps =  1/\Theta(\log n)$, with high probability Algorithm~\ref{alg:MQC} solves Problem~\ref{prob:MQC} in time $\bigOt{ \mathrm{poly}(n)\kappa(M)}$ .
\end{proof}

Compared with Chen and Gao's~\cite{ChenGao2017} algorithm, there are
two differences with Algorithm~\ref{alg:MQC}. First, the size of the
Boolean Macaulay matrix in Algorithm~\ref{alg:MQC} is $m2^n\times
2^n$, which leads to a smaller
lower bound of the \tqlscn and leaves a possibility of superpolynomial speedup
using Algorithm~\ref{alg:MQC}. By contrast, the size of the Macaulay matrix in Chen and Gao's
algorithm is $(m+n)(3n+1)^n\times (3n+1)^n$ \footnote{The parameters $m,n$ comes from Problem \ref{prob:MQC}}, which leads to a larger lower bound
of the \tqlscn that prohibits a potential quantum speedup.
Second, in Algorithm~\ref{alg:MQC}, the polynomial system have a unique solution, 
so in contrast to~\cite{ChenGao2017} the Boolean Macaulay linear system stays the same for every iteration and the number of iterations (measurements) required to obtain the solution of the polynomial system is $\bigO{\log n}$. However, the Valiant-Vazirani reduction needs $\bigO{n}$ iterations to generate a polynomial system that has a unique solution with high probability. This amounts to $\bigO{n \log n}$ iterations in total to find a solution. On the other hand, in Chen and Gao's algorithm, the polynomial system could have any finite number of solutions, so the Macaulay linear system needs to be updated after each iteration (measurement) and the number of iterations is $\bigO{n}$. 

%{\color{blue}  }
\section{Discussion}

The Boolean Macaulay linear system approach is an interesting framework to study giving insights to the limitations and capabilities of quantum computation. On one hand, a lot of problems such as Factoring, Graph isomorphism, and Learning with binary errors can be put into this single framework. On the other hand, the QLS algorithm used for the (Boolean) Macaulay linear system is $\mathsf{BQP}$-complete and the Factoring problem is known to be inside BQP by Shor's algorithm, therefore, if we can find an approach to get around the curse of the condition number of the Boolean Macaulay linear system derived from the Factoring problem, then we might be able to extend the result to other problems, such as Graph Isomorphism and Learning with binary errors, revealing new capabilities of quantum computation. 

Our analytical lower bound on the condition number decreases when there are multiple solutions of the polynomial systems, but the polynomial systems used for cryptography purpose usually have one or few solutions~\cite{caminata2017solving}, so our result gives strong evidence that the QLS algorithm cannot be used for attacking cryptosystems via the Macaulay matrix approach. Also, we suspect that having many solutions will not make the QLS algorithm to work substantially better. For example consider adding $l$ new field equations $y_i^2-y_i=0$, then the number of solution of the new polynomial system will increase by a factor of $2^l$, however the length of the shortest vector stays the same -- indeed one can see that the shortest vector is an affine combination of solutions where all the new variables are set to $0$.\footnote{To see this consider the coordinates corresponding to monomials not including the new variables.} %Therefore, one cannot compromise the lower bound result on the condition number by simply adding new variables.

Given an ill-conditioned QLSP, two main approaches have been proposed, one is the truncated QLS algorithm, the other one is the preconditioned QLS algorithm \cite{harrow2009QLinSysSolver}. Our lower bound on the \tqlscn prohibits further speedup by the truncated QLS algorithm, however further investigation is needed regarding the possibility of using preconditioned QLS algorithms, such as parallel sparse approximate inverse preconditioner~\cite{clader2013preconditioned}, circulant preconditioner~\cite{shao2018quantum}, fast inversion~\cite{tong2020PrecQLinSysSolvAndMatFunEval}, or develop new preconditioned QLS algorithms for the Boolean Macaulay linear system. A promising feature of the Boolean Macaulay linear system is that it cannot be de-quantized by known classical techniques \cite{chia2019SampdSubLinLowRankFramework} since the Boolean Macaulay linear system has high (full) column rank. 

The introduced variant of the quantum coupon collector problem provides an example of how to extract the solution efficiently if the values in the solution vector of a QLSP are correlated in a nice pattern. Since applications of the QLS algorithm usually gain restricted access to the solution vector, a generalization of our extraction method, utilizing more general correlated patterns, could have interesting applications.

When the Hamming weight of the solution of a Boolean polynomial system is logarithmic in the number of variables, our lower bound does not rule out a superpolynomial speedup over exhaustive search. Finding a real application that exhibits such a superpolynomial speedup would be very interesting. However, for applications like polynomial systems over a finite field, the employed reductions usually increase the number of variables in the corresponding polynomial system significantly -- likely compromising the potential speedup.

\section*{Acknowledgements}
A.G.\ thanks Rachel Player for inspiring discussions. J.L.\ would like to thank Eric R Anschuetz, Yuan Su, Yu-Ao Chen, Xiao-Shan Gao, Sevag Gharibian, Antonio Blanca, Eunou Lee, Mahdi Belbasi, Mingming Chen, and Rachel Player for helpful discussions and conversations.

Part of this work was done while the authors visited the Simons Institute for the Theory of Computing; we gratefully acknowledge its hospitality.

J.D.\ acknowledges support from NSF grant SaTC-1814221 and Taft Foundation. 
V.G.\ acknowledges support from NSERC and CIFAR; IQC is supported in part by the Government of Canada and the Province of Ontario.
A.G.\ acknowledges funding provided by Samsung Electronics Co., Ltd., for the project ``The Computational Power of Sampling on Quantum Computers'', by the Institute for Quantum Information and Matter, an NSF Physics Frontiers Center (NSF Grant PHY-1733907), and also by the EU's Horizon 2020 Marie Skłodowska-Curie program 891889-QuantOrder.
S.H.\ was partially supported by National Science Foundation awards CCF-1618287, CNS-1617802, and CCF-1617710, and by a Vannevar Bush Faculty Fellowship from the US Department of Defense.

\bibliographystyle{alphaUrlePrint} 
\bibliography{FACbib,hhl,qc_gily}

\newcommand{\etalchar}[1]{$^{#1}$}
\newcommand{\lName}{1}\newcommand{\arxiv}[1]{arXiv:
  \href{https://arxiv.org/abs/#1}{\ttfamily{#1}}\removefirstdot}\newcommand{\arXiv}[1]{arXiv:
  \href{https://arxiv.org/abs/#1}{\ttfamily{#1}}\removefirstdot}\def\removefirstdot#1{\if.#1{}\else#1\fi}\providecommand{\multiletter}[1]{#1}\renewcommand{\multiletter}[1]{#1}\DeclareRobustCommand{\dutchPrefix}[2]{#2}\providecommand{\dutchPrefix}[2]{#2}\renewcommand{\dutchPrefix}[2]{#2}\newcommand{\skp}[3]{#2}\newcommand{\focs
  }[1]{\if\lName1\skp{ }{Proceedings of the #1 {IEEE} Symposium on Foundations
  of Computer Science ({FOCS})}{ }\else{FOCS}\fi}\newcommand{\stoc
  }[1]{\if\lName1\skp{ }{Proceedings of the #1 {ACM} Symposium on the Theory of
  Computing ({STOC})}{ }\else{STOC}\fi}\newcommand{\soda }[1]{\if\lName1\skp{
  }{Proceedings of the #1 {ACM-SIAM} Symposium on Discrete Algorithms
  ({SODA})}{ }\else{SODA}\fi}\newcommand{\stacs }[1]{\if\lName1\skp{
  }{Proceedings of the #1 Symposium on Theoretical Aspects of Computer Science
  ({STACS})}{ }\else{STACS}\fi}\newcommand{\itcs }[1]{\if\lName1\skp{
  }{Proceedings of the #1 Innovations in Theoretical Computer Science
  Conference ({ITCS})}{ }\else{ITCS}\fi}\newcommand{\fsttcs
  }[1]{\if\lName1\skp{ }{Proceedings of the #1 International Conference on
  Foundations of Software Technology and Theoretical Computer Science
  ({FSTTCS})}{ }\else{FSTTCS}\fi}\newcommand{\mfcs }[1]{\if\lName1\skp{
  }{Proceedings of the #1 International Symposium on Mathematical Foundations
  of Computer Science ({MFCS})}{ }\else{MFCS}\fi}\newcommand{\ccc
  }[1]{\if\lName1\skp{ }{Proceedings of the #1 {IEEE} Conference on
  Computational Complexity ({CCC})}{ }\else{CCC}\fi}\newcommand{\isit
  }[1]{\if\lName1\skp{ }{Proceedings of the #1 {IEEE} International Symposium
  on Information Theory ({ISIT})}{ }\else{ISIT}\fi}\newcommand{\colt
  }[1]{\if\lName1\skp{ }{Proceedings of the #1 Conference On Learning Theory
  ({COLT})}{ }\else{COLT}\fi}\newcommand{\nips }[1]{\if\lName1\skp{ }{Advances
  in Neural Information Processing Systems #1 ({NIPS})}{
  }\else{NIPS}\fi}\newcommand{\aistats }[1]{\if\lName1\skp{ }{Proceedings of
  the #1 International Conference on Artificial Intelligence and Statistics
  ({AISTATS})}{ }\else{AISTATS}\fi}\newcommand{\icml }[1]{\if\lName1\skp{
  }{Proceedings of the #1 International Conference on Machine Learning
  ({ICML})}{ }\else{ICML}\fi}\newcommand{\icalp }[1]{\if\lName1\skp{
  }{Proceedings of the #1 International Colloquium on Automata, Languages, and
  Programming ({ICALP})}{ }\else{ICALP}\fi}\newcommand{\esa
  }[1]{\if\lName1\skp{ }{Proceedings of the #1 Annual European Symposium on
  Algorithms ({ESA})}{ }\else{ESA}\fi}\newcommand{\tqc }[1]{\if\lName1\skp{
  }{Proceedings of the #1 Conference on the Theory of Quantum Computation,
  Communication, and Cryptography ({TQC})}{}\else{TQC}\fi}\newcommand{\isaac
  }[1]{\if\lName1\skp{ }{Proceedings of the #1 International Symposium on
  Algorithms and Computation ({ISAAC})}{ }\else{ISAAC}\fi}\newcommand{\jacm
  }{\if\lName1\skp{ }{Journal of the ACM}{ }\else{J. ACM}\fi}\newcommand{\acmta
  }{\if\lName1\skp{ }{ACM Transactions on Algorithms}{ }\else{{ACM} Tr.
  Alg}\fi}\newcommand{\acmtct }{\if\lName1\skp{ }{ACM Transactions on
  Computation Theory}{ }\else{ACM Tr. Comp. Th.}\fi}\newcommand{\acmtqc
  }{\if\lName1\skp{ }{ACM Transactions on Quantum Computing}{ }\else{ACM Tr.
  Quant. Comp.}\fi}\newcommand{\jams }{\if\lName1\skp{ }{Journal of the AMS}{
  }\else{J. AMS}\fi}\newcommand{\pams }{\if\lName1\skp{ }{Proceedings of the
  AMS}{ }\else{Proc. AMS}\fi}\newcommand{\linalgappl }{\if\lName1\skp{ }{Linear
  Algebra and its Applications}{ }\else{Lin. Alg. \&
  App.}\fi}\newcommand{\jalgo }{\if\lName1\skp{ }{Journal of Algorithms}{
  }\else{J. Alg.}\fi}\newcommand{\jcss }{\if\lName1\skp{ }{Journal of Computer
  and System Sciences}{ }\else{J. Comp. Sys. Sci.}\fi}\newcommand{\cc
  }{\if\lName1\skp{ }{Computational Complexity}{ }\else{Comp.
  Comp.}\fi}\newcommand{\algor }{\if\lName1\skp{ }{Algorithmica}{
  }\else{Alg.}\fi}\newcommand{\comb }{\if\lName1\skp{ }{Combinatorica}{
  }\else{Comb.}\fi}\newcommand{\cacm }{\if\lName1\skp{ }{Communications of the
  ACM}{ }\else{Comm. ACM}\fi}\newcommand{\sigart }{\if\lName1\skp{ }{SIGART
  Bulletin}{ }\else{SIGART Bull.}\fi}\newcommand{\sigactn }{\if\lName1\skp{
  }{SIGACT News}{ }\else{SIGACT News}\fi}\newcommand{\eatcsbul
  }{\if\lName1\skp{ }{Bulletin of the {EATCS}}{ }\else{Bull.
  {EATCS}}\fi}\newcommand{\siamrev }{\if\lName1\skp{ }{SIAM Review}{
  }\else{SIAM Rev.}\fi}\newcommand{\siamjc }{\if\lName1\skp{ }{SIAM Journal on
  Computing}{ }\else{SIAM J. Comp.}\fi}\newcommand{\siamjo }{\if\lName1\skp{
  }{SIAM Journal on Optimization}{ }\else{SIAM J. Opt.}\fi}\newcommand{\siamjdm
  }{\if\lName1\skp{ }{SIAM Journal on Discrete Mathematics}{ }\else{SIAM J.
  Disc. Math.}\fi}\newcommand{\siamjnum }{\if\lName1\skp{ }{SIAM Journal on
  Numerical Analysis}{ }\else{SIAM J. Num. Anal.}\fi}\newcommand{\siamjmathanal
  }{\if\lName1\skp{ }{SIAM Journal on Mathematical Analysis}{ }\else{SIAM J.
  Math. Anal.}\fi}\newcommand{\discmath }{\if\lName1\skp{ }{Discrete
  Mathematics}{ }\else{Disc. Math.}\fi}\newcommand{\das }{\if\lName1\skp{
  }{Discrete Applied Mathematics}{ }\else{Disc. App.
  Math.}\fi}\newcommand{\amatstat }{\if\lName1\skp{ }{Annals of Mathematical
  Statistics}{ }\else{Ann. Math. Stat.}\fi}\newcommand{\rms }{\if\lName1\skp{
  }{Russian Mathematical Surveys}{ }\else{Russ. Math.
  Surv.}\fi}\newcommand{\invmath }{\if\lName1\skp{ }{Inventiones Mathematicae}{
  }\else{Inv. Math.}\fi}\newcommand{\jnumber }{\if\lName1\skp{ }{Journal of
  Number Theory}{ }\else{J. Num. Th.}\fi}\newcommand{\tcs }{\if\lName1\skp{
  }{Theoretical Computer Science}{ }\else{Theor. Comput.
  Sci.}\fi}\newcommand{\toc }{\if\lName1\skp{ }{Theory of Computing}{
  }\else{Th. Comp.}\fi}\newcommand{\cjtcs }{\if\lName1\skp{ }{Chicago Journal
  of Theoretical Computer Science}{}\else{Chic. J. Th. Comp.
  Sci.}\fi}\newcommand{\quantum }{\if\lName1\skp{ }{{Quantum}}{
  }\else{Quant.}\fi}\newcommand{\cmp }{\if\lName1\skp{ }{Communications in
  Mathematical Physics}{ }\else{Comm. Math. Phys.}\fi}\newcommand{\jmp
  }{\if\lName1\skp{ }{Journal of Mathematical Physics}{ }\else{J. Math.
  Phys.}\fi}\newcommand{\rspa }{\if\lName1\skp{ }{Proceedings of the Royal
  Society A}{ }\else{Proc. Roy. Soc. A}\fi}\newcommand{\qic }{\if\lName1\skp{
  }{Quantum Information and Computation}{ }\else{Quant. Inf. \&
  Comp.}\fi}\newcommand{\physrev }{\if\lName1\skp{ }{Physical Review}{
  }\else{Phys. Rev.}\fi}\newcommand{\pra }{\if\lName1\skp{ }{Physical Review
  A}{ }\else{Phys. Rev. A}\fi}\newcommand{\prb }{\if\lName1\skp{ }{Physical
  Review B}{ }\else{Phys. Rev. B}\fi}\newcommand{\pre }{\if\lName1\skp{
  }{Physical Review E}{ }\else{Phys. Rev. E}\fi}\newcommand{\prr
  }{\if\lName1\skp{ }{Physical Review Research}{ }\else{Phys. Rev.
  Research}\fi}\newcommand{\prx }{\if\lName1\skp{ }{Physical Review X}{
  }\else{Phys. Rev. X}\fi}\newcommand{\prl }{\if\lName1\skp{ }{Physical Review
  Letters}{ }\else{Phys. Rev. Lett.}\fi}\newcommand{\njp }{\if\lName1\skp{
  }{New Journal of Physics}{ }\else{New J. Phys.}\fi}\newcommand{\prapp
  }{\if\lName1\skp{ }{Physical Review Applied}{ }\else{Phys. Rev.
  Appl.}\fi}\newcommand{\physrep }{\if\lName1\skp{ }{Physics Reports}{
  }\else{Phys. Rep.}\fi}\newcommand{\rmp }{\if\lName1\skp{ }{Reviews of Modern
  Physics}{ }\else{Rev. Mod. Phys. }\fi}\newcommand{\phystoday
  }{\if\lName1\skp{ }{Physics Today}{ }\else{Phys.
  Today}\fi}\newcommand{\physics }{\if\lName1\skp{ }{Physics}{
  }\else{Phys.}\fi}\newcommand{\nature }{\if\lName1\skp{ }{Nature}{
  }\else{Nat.}\fi}\newcommand{\natcomm }{\if\lName1\skp{ }{Nature
  Communications}{ }\else{Nat. Comm.}\fi}\newcommand{\natphys }{\if\lName1\skp{
  }{Nature Physics}{ }\else{Nat. Phys.}\fi}\newcommand{\npjqi }{\if\lName1\skp{
  }{npj Quantum Information}{ }\else{npj Quant. Inf.}\fi}\newcommand{\scirep
  }{\if\lName1\skp{ }{Scientific Reports}{ }\else{Sci.
  Rep.}\fi}\newcommand{\science }{\if\lName1\skp{ }{Science}{
  }\else{Sci.}\fi}\newcommand{\jpa }{\if\lName1\skp{ }{Journal of Physics A:
  Mathematical and Theoretical}{ }\else{J. Phys. A}\fi}\newcommand{\ijtp
  }{\if\lName1\skp{ }{International Journal of Theoretical Physics}{
  }\else{Int. J. Th. Phys.}\fi}\newcommand{\jmo }{\if\lName1\skp{ }{Journal of
  Modern Optics}{ }\else{J. Mod. Opt.}\fi}\newcommand{\jstatph
  }{\if\lName1\skp{ }{Journal of Statistical Physics}{ }\else{J. Stat.
  Phys.}\fi}\newcommand{\pnas }{\if\lName1\skp{ }{Proceedings of the National
  Academy of Sciences}{ }\else{PNAS}\fi}\newcommand{\lncs }{\if\lName1\skp{
  }{Lecture Notes in Computer Science}{ }\else{L. Notes Comp.
  Sci.}\fi}\newcommand{\lnai }{\if\lName1\skp{ }{Lecture Notes in Artificial
  Intelligence}{ }\else{L. Notes Art. Int.}\fi}\newcommand{\lnm
  }{\if\lName1\skp{ }{Lecture Notes in Mathematics}{ }\else{L. Notes
  Math.}\fi}\newcommand{\tams }{\if\lName1\skp{ }{Transactions of the American
  Mathematical Society}{ }\else{Trans. AMS}\fi}\newcommand{\ieeetit
  }{\if\lName1\skp{ }{{IEEE} Transactions on Information Theory}{ }\else{{IEEE}
  Trans. Inf. Th.}\fi}\newcommand{\iscs }{\if\lName1\skp{ }{International
  Series in Computer Science}{ }\else{Int. Ser. Comp.
  Sci.}\fi}\newcommand{\tocl }{\if\lName1\skp{ }{Theory of Computing Library}{
  }\else{Th. Comp. Lib.}\fi}
\begin{thebibliography}{AOAGC18}

\bibitem[Aar15]{aaronson2015ReadTheFinePrint}
Scott Aaronson.
\newblock \href{http://dx.doi.org/10.1038/nphys3272}{Read the fine print}.
\newblock {\em \natphys}, 11(4):291--293, 2015.

\bibitem[ABC{\etalchar{+}}20]{arunachalam2020QuantumCouponCollector}
Srinivasan Arunachalam, Aleksandrs Belovs, Andrew~M. Childs, Robin Kothari,
  Ansis Rosmanis, and Ronald de~Wolf.
\newblock \href{http://dx.doi.org/10.4230/LIPIcs.TQC.2020.10}{Quantum coupon
  collector}.
\newblock In {\em \tqc{15th}}, pages 10:1--10:17, 2020.
\newblock \arxiv{2002.07688}.

\bibitem[AFI{\etalchar{+}}04]{ars2004comparison}
Gw{\'e}nol{\'e} Ars, Jean-Charles Faugere, Hideki Imai, Mitsuru Kawazoe, and
  Makoto Sugita.
\newblock
  \href{http://dx.doi.org/https://doi.org/10.1007/978-3-540-30539-2_24}{Comparison
  between {XL} and {G}r{\"o}bner basis algorithms}.
\newblock In {\em International Conference on the Theory and Application of
  Cryptology and Information Security}, pages 338--353. Springer, 2004.

\bibitem[Amb12]{ambainis2010VTAA}
Andris Ambainis.
\newblock \href{http://dx.doi.org/10.4230/LIPIcs.STACS.2012.636}{Variable time
  amplitude amplification and quantum algorithms for linear algebra problems}.
\newblock In {\em \stacs{29th}}, pages 636--647, 2012.
\newblock \arxiv{1010.4458}.

\bibitem[AOAGC18]{anschuetz2018variational}
Eric~R Anschuetz, Jonathan~P Olson, Al{\'a}n Aspuru-Guzik, and Yudong Cao.
\newblock
  \href{http://dx.doi.org/https://doi.org/10.1007/978-3-030-14082-3_7}{Variational
  quantum factoring}.
\newblock \arxiv{1808.08927}, 2018.

\bibitem[B{\etalchar{+}}18]{buchberger2018grobner}
Bruno Buchberger et~al.
\newblock \href{http://dx.doi.org/10.2969/aspm/07710025}{Gr{\"o}bner bases
  computation by triangularizing {M}acaulay matrices}.
\newblock In {\em The 50th Anniversary of Gr{\"o}bner Bases}, pages 25--33.
  Mathematical Society of Japan, 2018.

\bibitem[Bat13]{batselier2013numerical}
Kim Batselier.
\newblock \href{http://dx.doi.org/10.13140/RG.2.1.3137.9608}{{\em A numerical
  linear algebra framework for solving problems with multivariate
  polynomials}}.
\newblock PhD thesis, KU Leuven (Leuven, Belgium), 2013.

\bibitem[BBHT98]{boyer1998TightBoundsOnQuantumSearching}
Michel Boyer, Gilles Brassard, Peter H{\o}yer, and Alain Tapp.
\newblock
  \href{http://dx.doi.org/10.1002/(SICI)1521-3978(199806)46:4/5<493::AID-PROP493>3.0.CO;2-P}{Tight
  bounds on quantum searching}.
\newblock {\em Fortschritte der Physik}, 46(4--5):493--505, 1998.
\newblock \arxiv{quant-ph/9605034}.

\bibitem[BFSS13]{bardet2013complexity}
Magali Bardet, Jean-Charles Faug{\`e}re, Bruno Salvy, and Pierre-Jean
  Spaenlehauer.
\newblock \href{http://dx.doi.org/https://doi.org/10.1016/j.jco.2012.07.001}{On
  the complexity of solving quadratic boolean systems}.
\newblock {\em Journal of Complexity}, 29(1):53--75, 2013.

\bibitem[BKW19]{bjorklund2019solving}
Andreas Bj{\"o}rklund, Petteri Kaski, and Ryan Williams.
\newblock \href{http://dx.doi.org/10.4230/LIPIcs.ICALP.2019.26}{Solving systems
  of polynomial equations over {GF}(2) by a parity-counting self-reduction}.
\newblock In {\em \icalp{46th}}, 2019.

\bibitem[Bur02]{burges2002factoring}
Christopher~JC Burges.
\newblock Factoring as optimization.
\newblock {\em Microsoft Research MSR-TR-200}, 2002.

\bibitem[BV97]{bernstein1997quantum}
Ethan Bernstein and Umesh Vazirani.
\newblock
  \href{http://dx.doi.org/https://doi.org/10.1137/S0097539796300921}{Quantum
  complexity theory}.
\newblock {\em SIAM Journal on computing}, 26(5):1411--1473, 1997.

\bibitem[BY18]{bernstein2018asymptotically}
Daniel~J Bernstein and Bo-Yin Yang.
\newblock
  \href{http://dx.doi.org/https://doi.org/10.1007/978-3-319-79063-3_23}{Asymptotically
  faster quantum algorithms to solve multivariate quadratic equations}.
\newblock In {\em International Conference on Post-Quantum Cryptography}, pages
  487--506. Springer, 2018.

\bibitem[CG17]{caminata2017solving}
Alessio Caminata and Elisa Gorla.
\newblock
  \href{http://dx.doi.org/https://doi.org/10.1007/978-3-030-68869-1_1}{Solving
  multivariate polynomial systems and an invariant from commutative algebra}.
\newblock \arXiv{1706.06319}, 2017.

\bibitem[CG21]{ChenGao2017}
Yu-Ao Chen and Xiao-Shan Gao.
\newblock \href{http://dx.doi.org/10.1007/s11424-020-0028-6}{Quantum algorithm
  for {B}oolean equation solving and quantum algebraic attack on
  cryptosystems}.
\newblock {\em Journal of Systems Science and Complexity}, 2021.
\newblock \arxiv{1712.06239}.

\bibitem[CGJ19]{chakraborty2018BlockMatrixPowers}
Shantanav Chakraborty, András Gilyén, and Stacey Jeffery.
\newblock \href{http://dx.doi.org/10.4230/LIPIcs.ICALP.2019.33}{The power of
  block-encoded matrix powers: {I}mproved regression techniques via faster
  {H}amiltonian simulation}.
\newblock In {\em \icalp{46th}}, pages 33:1--33:14, 2019.
\newblock \arxiv{1804.01973}.

\bibitem[CGL{\etalchar{+}}20]{chia2019SampdSubLinLowRankFramework}
Nai-Hui Chia, Andr\'{a}s Gily\'{e}n, Tongyang Li, Han-Hsuan Lin, Ewin Tang, and
  Chunhao Wang.
\newblock \href{http://dx.doi.org/10.1145/3357713.3384314}{Sampling-based
  sublinear low-rank matrix arithmetic framework for dequantizing quantum
  machine learning}.
\newblock In {\em \stoc{52nd}}, page 387–400, 2020.
\newblock \arxiv{1910.06151}.

\bibitem[CGY18]{ChenGao2018}
Yu-Ao Chen, Xiao-Shan Gao, and Chun-Ming Yuan.
\newblock Quantum algorithm for optimization and polynomial system solving over
  finite field and application to cryptanalysis, 2018.
\newblock \arxiv{1802.03856}.

\bibitem[CJS13]{clader2013preconditioned}
B~David Clader, Bryan~C Jacobs, and Chad~R Sprouse.
\newblock
  \href{http://dx.doi.org/https://doi.org/10.1103/PhysRevLett.110.250504}{Preconditioned
  quantum linear system algorithm}.
\newblock {\em Physical review letters}, 110(25):250504, 2013.

\bibitem[CKPS04]{courtois2000efficient}
Nicolas Courtois, Alexander Klimov, Jacques Patarin, and Adi Shamir.
\newblock
  \href{http://dx.doi.org/https://doi.org/10.1007/3-540-45539-6_27}{Efficient
  algorithms for solving overdefined systems of multivariate polynomial
  equations}.
\newblock In {\em International Conference on the Theory and Applications of
  Cryptographic Techniques}, pages 392--407. Springer, 2000 (Extended version
  as of 24 Aug, 2004.
\newblock \url{ http://www.minrank.org/xlfull.pdf)}.

\bibitem[CKS17]{childs2015QLinSysExpPrec}
Andrew~M. Childs, Robin Kothari, and Rolando~D. Somma.
\newblock \href{http://dx.doi.org/10.1137/16M1087072}{Quantum algorithm for
  systems of linear equations with exponentially improved dependence on
  precision}.
\newblock {\em \siamjc}, 46(6):1920--1950, 2017.
\newblock \arxiv{1511.02306}.

\bibitem[Die04]{diem2004xl}
Claus Diem.
\newblock
  \href{http://dx.doi.org/https://doi.org/10.1007/978-3-540-30539-2_23}{The
  {XL}-algorithm and a conjecture from commutative algebra}.
\newblock In {\em International Conference on the Theory and Application of
  Cryptology and Information Security}, pages 323--337. Springer, 2004.

\bibitem[DS13]{ding2013solving}
Jintai Ding and Dieter Schmidt.
\newblock
  \href{http://dx.doi.org/https://doi.org/10.1007/978-3-642-42001-6_4}{Solving
  degree and degree of regularity for polynomial systems over a finite fields}.
\newblock In {\em Number Theory and Cryptography}, pages 34--49. Springer,
  2013.

\bibitem[FHK{\etalchar{+}}17]{faugere2017fast}
Jean-Charles Faugere, Kelsey Horan, Delaram Kahrobaei, Marc Kaplan, Elham
  Kashefi, and Ludovic Perret.
\newblock Fast quantum algorithm for solving multivariate quadratic equations.
\newblock \arXiv{1712.07211}, 2017.

\bibitem[GSLW18]{gilyen2018QSingValTransfArxiv}
András Gilyén, Yuan Su, Guang~Hao Low, and Nathan Wiebe.
\newblock Quantum singular value transformation and beyond: {E}xponential
  improvements for quantum matrix arithmetics [full version], 2018.
\newblock \arxiv{1806.01838}.

\bibitem[GSLW19]{gilyen2018QSingValTransf}
András Gilyén, Yuan Su, Guang~Hao Low, and Nathan Wiebe.
\newblock \href{http://dx.doi.org/10.1145/3313276.3316366}{Quantum singular
  value transformation and beyond: {E}xponential improvements for quantum
  matrix arithmetics}.
\newblock In {\em \stoc{51st}}, pages 193--204, 2019.
\newblock \arxiv{1806.01838}.

\bibitem[HHL09]{harrow2009QLinSysSolver}
Aram~W. Harrow, Avinatan Hassidim, and Seth Lloyd.
\newblock \href{http://dx.doi.org/10.1103/PhysRevLett.103.150502}{Quantum
  algorithm for linear systems of equations}.
\newblock {\em \prl}, 103(15):150502, 2009.
\newblock \arxiv{0811.3171}.

\bibitem[Juk11]{jukna2011ExtremalCombi2}
Stasys Jukna.
\newblock \href{http://dx.doi.org/10.1007/978-3-642-17364-6}{{\em Extremal
  Combinatorics - With Applications in Computer Science (2nd ed.)}}.
\newblock Texts in Theoretical Computer Science. Springer, 2011.

\bibitem[KP17]{kerenidis2016QRecSys}
Iordanis Kerenidis and Anupam Prakash.
\newblock \href{http://dx.doi.org/10.4230/LIPIcs.ITCS.2017.49}{Quantum
  recommendation systems}.
\newblock In {\em \itcs{8th}}, pages 49:1--49:21, 2017.
\newblock \arxiv{1603.08675}.

\bibitem[LT20]{lin2019OptimalQEigenstateFiltering}
Lin Lin and Yu~Tong.
\newblock \href{http://dx.doi.org/10.22331/q-2020-11-11-361}{Optimal polynomial
  based quantum eigenstate filtering with application to solving quantum linear
  systems}.
\newblock {\em \quantum}, 4:361, 2020.
\newblock \arxiv{1910.14596}.

\bibitem[Per16]{perret2016bases}
Ludovic Perret.
\newblock {\em Bases de Gr{\"o}bner en Cryptographie Post-Quantique}.
\newblock PhD thesis, UPMC-Paris 6 Sorbonne Universit{\'e}s, 2016.

\bibitem[Rob55]{robbins1955RemarkOnStrilingsFormula}
Herbert Robbins.
\newblock \href{http://dx.doi.org/10.2307/2308012}{A remark on {S}tirling's
  formula}.
\newblock {\em The American Mathematical Monthly}, 62(1):26--29, 1955.

\bibitem[Src22]{SrcWolframNotebook2022}
The {M}athematica source code is also available at the {W}oflram {N}otebook
  {A}rchive: \url{https://notebookarchive.org/2022-02-1ec5yyv}, 2022.

\bibitem[SSO19]{subasi2019QAlgSysLinEqsAdiabatic}
Yi\u{g}it Suba\c{s}\i, Rolando~D. Somma, and Davide Orsucci.
\newblock \href{http://dx.doi.org/10.1103/PhysRevLett.122.060504}{Quantum
  algorithms for systems of linear equations inspired by adiabatic quantum
  computing}.
\newblock {\em \prl}, 122(6):060504, 2019.
\newblock \arxiv{1805.10549}.

\bibitem[SX18]{shao2018quantum}
Changpeng Shao and Hua Xiang.
\newblock
  \href{http://dx.doi.org/https://doi.org/10.1103/PhysRevA.98.062321}{Quantum
  circulant preconditioner for a linear system of equations}.
\newblock {\em Physical Review A}, 98(6):062321, 2018.

\bibitem[TAWL20]{tong2020PrecQLinSysSolvAndMatFunEval}
Yu~Tong, Dong An, Nathan Wiebe, and Lin Lin.
\newblock
  \href{http://dx.doi.org/https://doi.org/10.1103/PhysRevA.104.032422}{Fast
  inversion, preconditioned quantum linear system solvers, and fast evaluation
  of matrix functions}.
\newblock \arxiv{2008.13295}, 2020.

\bibitem[VV86]{valiant1986NPEasyAsDetectingUniqueSols}
Leslie~G. Valiant and Vijay~V. Vazirani.
\newblock \href{http://dx.doi.org/10.1016/0304-3975(86)90135-0}{{NP} is as easy
  as detecting unique solutions}.
\newblock {\em \tcs}, 47:85--93, 1986.
\newblock Earlier version in STOC'85.

\bibitem[{\dutchPrefix{Wolf}{d}}W19]{wolf2019QCLectureNotes}
Ronald {\dutchPrefix{Wolf}{d}}e~Wolf.
\newblock Quantum computing: Lecture notes, 2019.
\newblock \arxiv{1907.09415}.

\bibitem[WW15]{wiesinger2015grobner}
Manuela Wiesinger-Widi.
\newblock \href{http://dx.doi.org/https://doi.org/10.1145/2016567.2016594}{{\em
  Gr{\"o}bner bases and generalized sylvester matrices}}.
\newblock PhD thesis, Johannes Kepler University Linz, Austria, 2015.

\end{thebibliography}

\appendix
\section{Simple proof of the unique solution case} \label{append:simple}
Here we present a simple proof for the correctness of Algorithm
\ref{alg:MQC} for Problem \ref{prob:MQC} when it has a unique
solution. Let $\s=(\s_1, \s_2,\dots \s_n) \in \{0,1 \}^n$ be the unique solution 
of a set of polynomials $\calF$.  
Let $\hat{y}=[\s_1, \s_2,\dots, \s_i \s_j,\dots , \prod_{i=1}^{n}\s_i]^{\top}$ be
the $0/1$ solution vector labeled by the multilinear monomials under the assignment $\s$.
 
Next, we will show that the Boolean Macaulay linear system $M\svec = \bvec$ has the
unique solution $\hat{y}$ when $\calF$ has the unique solution $\s$. 
In this case, we have $\hat{y}=M^{+}\bvec$ because the matrix $M$ has linearly independent columns. 
When $\calF$ has more than one solution, the columns of the matrix $M$ are not linearly independent and the solutions of $M\svec = \bvec$ form a multidimensional affine subspace.

\begin{lemma}\label{lem:idealunique}~\cite[Theorem 2]{ars2004comparison}
	If a set of polynomials
	$\calF \subseteq \mathbb{C}[x_1,\dots, x_n]$ has a unique solution
	$\s= (\s_1, \s_2, \dots, \s_n)$, then the following two polynomial ideals coincide
	\[\left\langle \calF \right\rangle = \left\langle x_1 - \s_1, x_2 -
	\s_2, \dots, x_n - \s_n \right\rangle.\]
\end{lemma}

\begin{theorem} \label{thm: uniquesolution}
  Given a set of polynomials $\calF = \calF_1 \cup \calF_{2} \subseteq \mathbb{C}[x_1,\dots, x_n]$, where
  $\calF_1=\{f_1, f_2,\dots, f_m\}$ and $\calF_2=\{x_1^2-x_1, x_2^2-x_2,\dots,
  x_n^2-x_n\}$. Suppose  $\calF$ has a unique solution $\s=(\s_1,\dots,\s_n)$, where $\calF_2$ forces the root of the set of polynomials $\calF_1$ to be Boolean. Let $\hat{y}$ be the multilinear monomial solution vector corresponding to the solution $\s$,
  then the Boolean Macaulay linear system $M\svec=\bvec$ of total degree $n$ has the unique solution $\hat{y}$ $=M^{+}\bvec$.
\end{theorem}

\begin{proof}
		First, we prove that for all the 
	nontrivial multilinear monomials $X^{\beta}$,  
	the polynomial $X^{\beta} - \prod_{i=1}^{n} \s_i^{\beta_i} =\prod_{i=1}^{n} x_i^{\beta_i} - \prod_{i=1}^{n} \s_i^{\beta_i}$
	is in $\left\langle \calF\right\rangle $, where $\beta \in
	\{0,1\}^n\backslash \{0^n\}$. The proof is by induction on the degree
	$d$.  For the base case $d=1$, by Lemma~\ref{lem:idealunique}, for all $1 \leq k \leq n$, 
	$ x_k - \s_k \in
	\left\langle \calF\right\rangle$. That is, for each $k$, there exist $p_i, q_j$ 
	such that $ x_k - \s_k  = \sum_{i=1}^{m} p_i f_i +\sum_{j=1}^{n} q_j (x_j^2-x_j) $.  
   Let $B_d=\{\text{all multilinear monomials with degree $d$}\}$, and suppose the claim is true for $d$, that is, for any 
	$ X^{\beta^{'}}_{d} \in B_d$, $ X^{\beta^{'}}_{d}-  \prod_{i=1}^{n} \s_i^{\beta^{'}_{i}} \in \left\langle \calF\right\rangle$.
	For any $X^\beta_{d+1} \in B_{d+1}$, there exists some $ x_k$ and $X^{\beta^{'}}_{d}$ such that 
	$X^\beta_{d+1}= x_k \cdot  X^{\beta^{'}}_{d} $ and $\prod_{i=1}^{n} \s_i^{\beta_i}= \s_k \cdot \prod_{i=1}^{n} \s_i^{\beta^{'}_{i}}$, 
 %{\color{blue} the notation $\beta$ is confusing, rewrite this paragraph}
 then 
	\[
	X^\beta_{d+1} -  \prod_{i=1}^{n} \s_i^{\beta_i} =  x_k  (X^{\beta^{'}}_{d} - \prod_{i=1}^{n} \s_i^{\beta^{'}_{i}}) +  (x_k - \s_k)\prod_{i=1}^{n} \s_i^{\beta^{'}_{i}} ,
	\] which implies that $ X^\beta_{d+1} -  \prod_{i=1}^{n} \s_i^{\beta_i} \in \left\langle \calF\right\rangle $. 
	Therefore,  for all  nontrivial multilinear monomials $X^{\beta}$, there exists $p_{i\beta}, q_{j\beta} \in \C[x_1,\dots, x_n]$  such that 
	$X^{\beta} -  \prod_{i=1}^{n} \s_i^{\beta_i} = \sum_{i=1}^{m} p_{i\beta} f_i +\sum_{j=1}^{n} q_{j\beta} (x_j^2-x_j)\in  \left\langle \calF\right\rangle 
	$, where $ i \in [m], j \in [n]$.
	
The Boolean Macaulay matrix
	$\left[\begin{array}{cc}
	M  & -\bvec
	\end{array}\right] 
 $ 
  is the augmented matrix of the Boolean Macaulay linear system $M\svec=\bvec$ of the
	set of polynomials $\calF$.
	Since $X^{\beta} -  \prod_{i=1}^{n} \s_i^{\beta_i} = \psi\left(X^{\beta} -  \prod_{i=1}^{n} \s_i^{\beta_i}\right) =  \sum_{i=1}^{m} \psi(p_{i\beta} f_i)$, and polynomial addition, subtraction, and multiplication in the polynomial ideal  $\left\langle \calF\right\rangle$ correspond to 
	row operations of the Boolean Macaulay matrix $\left[\begin{array}{cc}
	M & -\bvec
	\end{array}\right]$, we can perform row operations on the Boolean Macaulay matrix according to $\sum_{i=1}^{m} \psi(p_{i\beta} f_i)$.  By those row operations we can obtain an extended matrix of form 
	$\left[ \begin{array}{cc}
	M & -\bvec \\
	I & -\svec
	\end{array} \right]$, 
	where the columns of the identity matrix $I$ are indexed by the nontrivial multilinear monomials and entries of $\svec$ are values of $ \prod_{i=1}^{n} \s_i^{\beta_i}$. As performing row operations on a matrix does not change the matrix rank, the matrix $M$ must have full column rank.
	Therefore, the Boolean Macaulay linear system  has the unique solution $\hat{y}= M^+\bvec $. 
\end{proof}

\section{Bounds on binomial coefficients}\label{apx:binom}
In this appendix we derive some standard bounds on binomial coefficients for completeness. First we show that for $h\leq n/2$
\begin{equation}\label{eq:binomEntr}
\sum_{j=0}^{h}\binom{n}{j}\leq 3\sqrt{h}\binom{n}{h}.
\end{equation}
We use the following upper bound \cite[Corollary 22.9]{jukna2011ExtremalCombi2} on binomial coefficients 
\begin{align*}%\label{eq:entropyBinom}
\forall\, 0 < h \leq n/2 :\,\, \sum_{j=0}^h\binom{n}{j} &\leq 2^{n\cdot H(h/n)}\\
&=2^{n\left(-\frac{h}{n}\log_2\left(\frac{h}{n}\right)-\frac{n-h}{n}\log_2\left(\frac{n-h}{n}\right)\right)}\\
&=\Big(\frac{n}{h}\Big)^{\!h}\Big(\frac{n}{n-h}\Big)^{\!\!n-h},
\end{align*}
in combination with Stirling's approximation~\cite{robbins1955RemarkOnStrilingsFormula} $\sqrt{2\pi}\sqrt{n}\left(\frac{n}{e}\right)^{\!n}\leq n! \leq e \sqrt{n}\left(\frac{n}{e}\right)^{\!n}$, yielding
\begin{align*}
\binom{n}{h}=\frac{n!}{(n-h)!h!}
&\geq \frac{\sqrt{2\pi}\sqrt{n}\left(\frac{n}{e}\right)^{\!n}}
{e\sqrt{n-h}\left(\frac{n-h}{e}\right)^{\!n-h}e\sqrt{h}\left(\frac{h}{e}\right)^{\!h}}\\
&=\frac{\sqrt{2\pi}}{e^2}\sqrt{\frac{n}{n-h}}\frac{1}{\sqrt{h}}\Big(\frac{n}{h}\Big)^{\!h}\Big(\frac{n}{n-h}\Big)^{\!\!n-h}\\
&\geq\frac{1}{3\sqrt{h}}\sum_{j=0}^h\binom{n}{j}.
\end{align*}

Another bound that we use is 
\begin{equation}\label{eq:binomGeom}
\sum_{j=1}^{h}\binom{n}{i} \leq \binom{n}{h}\frac{n-h+1}{n-2h+1},
\end{equation}
which can be shown by the summation of an upper bound by a geometric series, i.e.,
\begin{align*}
 \sum_{i=1}^{h}\binom{n}{i}& = \binom{n}{h} (1 + \binom{n}{h-1}/\binom{n}{h} +  \binom{n}{h-2}/\binom{n}{h}\\&+\cdots+ \binom{n}{1}/\binom{n}{h}) \\
 &\leq \binom{n}{h} (1 +   h/(n-h+1) +    h^2/(n-h+1)^2\\&+\cdots+  h^{h-1}/(n-h+1)^{h-1})\\
 &\leq \binom{n}{h}\frac{1}{1-h/(n-h+1)}\\
 &= \binom{n}{h}\frac{n-h+1}{n-2h+1}.
 \end{align*}
 
%\section{Mathematica code for symbolically verifying our lower bound~\cite{SrcWolframNotebook2022}}\label{apx:Mathematica}
% \hspace*{-5mm}
 %\begin{mmaCell}[functionlocal={s,i,j},pattern={n_,n,d_,d,\#}]{Code}
%	GramMax[n_, d_]:= Table[Table[d^s Binomial[i, s]Binomial[j, s]/Binomial[n, s], 
 %				      {s, 1, Min[i, j]}] // Total, {i, 1, n}, {j, 1, n}];
% 	BoundMax[n_]:= 2GramMax[n, 3n] - Table[Min[i, j]^Min[i, j], {i, 1, n}, {j, 1, n}];
 %	ParallelMap[(Print[#]; #[[2]]) & @ Timing[{#, 
 %	    PositiveDefiniteMatrixQ[BoundMax[#]]}] &, Range[300], Method -> "FinestGrained"]
% \end{mmaCell}
 %\begin{mmaCell}{Output}
 %	\{0.000715, \{1,True\}\}
 %	\{0.000922, \{2,True\}\}
 %	\{0.000413, \{3,True\}\}
 %	\vdots 	
 %	\{36634.1, \{296, True\}\}                              
 %	\{37459.0, \{297, True\}\}                                         
 %	\{37858.5, \{298, True\}\}                                          
 %	\{38813.1, \{299, True\}\}                              
 %	\{39445.0, \{300, True\}\} 	
 %\end{mmaCell}

\end{document}